\documentclass[11pt]{article}
\usepackage{fullpage}
\usepackage{tikz}
\usepackage{amssymb,amsmath,amsthm}
\usepackage{xspace}
\usepackage{extras}
\usepackage[noend]{algorithmic}
\usepackage{algorithm}

\usepackage [pdftex, %dvipsone, %or dvips
colorlinks=true, 
linkcolor=webbrown,  %definedbelow 
filecolor=webbrown, %definedbelow 
citecolor=webgreen, %definedbelow
bookmarksopen=false, 
pdfpagemode=None]
{hyperref}
\definecolor{webgreen}{rgb}{0,.5,0} 
\definecolor{webbrown}{rgb}{.6,0,0}
{\makeatletter
 \gdef\xxxmark{%
   \expandafter\ifx\csname @mpargs\endcsname\relax % in minipage?
     \expandafter\ifx\csname @captype\endcsname\relax % in figure/caption?
       \marginpar{xxx}% not in a caption or minipage, can use marginpar
     \else
       xxx % notice trailing space
     \fi
   \else
     xxx % notice trailing space
   \fi}
 \gdef\xxx{\@ifnextchar[\xxx@lab\xxx@nolab}
 \long\gdef\xxx@lab[#1]#2{{\bf [\xxxmark #2 ---{\sc #1}]}}
 \long\gdef\xxx@nolab#1{{\bf [\xxxmark #1]}}
 \long\gdef\xxx@lab[#1]#2{}\long\gdef\xxx@nolab#1{}
}
\newcommand{\OPT}{\ensuremath{\mathsf{OPT}}}
\newcommand{\len}{\ensuremath{\mathsf{length}}}
\newcommand{\tw}{\ensuremath{\mathsf{tw}}}
\newcommand{\dist}{\ensuremath{\mathsf{dist}}}
\newcommand{\cost}{\ensuremath{\mathsf{cost}}}
\newcommand{\sfor}{\ensuremath{\mathsf{SteinerForest}}}
\newcommand{\D}{\ensuremath{\mathcal{D}}}
\newcommand{\dsat}{\ensuremath{\D^\mathsf{sat}}}
\newcommand{\dunsat}{\ensuremath{\D^\mathsf{unsat}}}
\newcommand{\C}{\ensuremath{\mathcal{C}}}
\newcommand{\eS}{\ensuremath{\mathcal{S}}}
\newcommand{\B}{\ensuremath{\mathcal{B}}}

\newcommand{\eps}{\ensuremath{\epsilon}}
\newcommand{\R}{\ensuremath{\mathbb{R}}}

\newcommand{\I}{\mathcal{I}}

\newcommand{\DP}{\ensuremath{\mathrm{DP}\xspace}}
\newcommand{\pcst}{\textsf{PCST}\xspace}
\newcommand{\pcs}{\textsf{PCS}\xspace}
\newcommand{\pctsp}{\textsf{PCTSP}\xspace}
\newcommand{\pcsf}{\textsf{PCSF}\xspace}
\newcommand{\spcsf}{\textsf{SPCSF}\xspace}
\newcommand{\spctsp}{\textsf{SPCTSP}\xspace}
\newcommand{\spcs}{\textsf{SPCS}\xspace}
\newcommand{\psubmodforest}{\textsf{Submodular Prize-Collecting Steiner Forest}\xspace}
\newcommand{\prob}[1]{\textsf{#1}}
\newcommand{\algo}[1]{\textsc{#1}}
\newtheorem{theorem}{Theorem}
\newtheorem{lemma}[theorem]{Lemma}

\newtheorem{corollary}[theorem]{Corollary}
\theoremstyle{definition}

\begin{document}

%%%%
\author{MohammadHossein Bateni\thanks{Department of Computer Science, Princeton University, Princeton, NJ 08540; Email: \textsf{mbateni@cs.princeton.edu}.
The author is also with Center for Computational Intractability, Princeton, NJ 08540.
He was supported by
a Gordon Wu fellowship as well as
NSF ITR grants
                      CCF-0205594, CCF-0426582 and NSF CCF 0832797,
                      NSF CAREER award CCF-0237113,
                      MSPA-MCS award 0528414,
                      NSF expeditions award 0832797.
}
 \and MohammadTaghi Hajiaghayi\thanks{AT\&T Labs--Research, Florham
   Park, NJ 07932; Email: \textsf{hajiagha@research.att.com}.}
\and
 D\'{a}niel Marx\thanks{The Blavatnik
School of Computer Science,
Tel Aviv University,
Tel Aviv, Israel; \texttt{dmarx@cs.bme.hu.}
He is supported by ERC
   Advanced Grant DMMCA.}
}
\title{Prize-collecting Network Design on Planar Graphs}
\date{}
\maketitle
%%%%
\begin{abstract}
In this paper, we reduce {\sf Prize-Collecting Steiner TSP
(\pctsp)}, {\sf Prize-Collecting Stroll (\pcs)}, {\sf
Prize-Collecting Steiner Tree (\pcst)}, {\sf Prize-Collecting
Steiner Forest (\pcsf)} and more generally {\sf Submodular
Prize-Collecting Steiner Forest (\spcsf)}  on planar graphs (and
more generally bounded-genus graphs) to the same problems on graphs
of bounded treewidth. More precisely, we show any $\alpha$-%
approximation algorithm for these problems on graphs of bounded
treewidth gives an $(\alpha+\epsilon)$-approximation algorithm for
these problems on planar graphs (and more generally bounded-genus
graphs), for any constant $\epsilon> 0$. Since \pcs, \pctsp, and
\pcst can be solved exactly on graphs of bounded treewidth using
dynamic programming, we obtain PTASs for these problems on planar
graphs and bounded-genus graphs. In contrast, we show \pcsf is 
APX-hard to approximate on series-parallel graphs, which are planar
graphs of treewidth at most 2. This result is interesting on its
own because it gives the first provable hardness separation between
prize-collecting and non-prize-collecting (regular) versions of the
problems: regular {\sf Steiner Forest} is known to
be polynomially solvable on series-parallel graphs and admits a PTAS on graphs of bounded
treewidth. An analogous hardness result can be shown for
Euclidian \pcsf. This ends the common belief that prize-collecting
variants should not add any new hardness to the problems.
\end{abstract}
%%%%

%\thispagestyle{empty}
\newpage

\section{Introduction}
Prize-collecting problems involve situations where there are various
demands that desire to be ``served'' by some structure and we must
find the structure of lowest cost to accomplish this.  However, if
some of the demands are too expensive to serve, then we can refuse
to serve them and instead pay a penalty. In particular,
prize-collecting Steiner problems are well-known
 network design problems with several applications in
expanding telecommunications networks (see for
example~\cite{JohnsonMP00-pcst,SCRS}), cost sharing, and Lagrangian relaxation
techniques (see e.g. \cite{JV01,CRW01}). A general form of these
problems is the {\sf Prize-Collecting Steiner Forest} (\pcsf)
problem\footnote{In the literature, this problem is also called {\sf
Prize-Collecting Generalized Steiner Tree}.}: given a network
(graph) $G=(V,E)$, a set of source-sink pairs\footnote{Source-sink pairs are sometimes called demands.}
$\D=\{\{s_1,t_1\},\{s_2,t_2\}, \dots,\{s_{k},t_{k}\}\}$, a
non-negative cost function $c:E\rightarrow \R_+$, and a
non-negative penalty function $\pi:2^\D\rightarrow \R_+$, our goal
is a minimum-cost way of installing (buying) a set of links (edges)
and paying the penalty for those pairs which are not connected via
installed links. We also consider the problem with a general penalty
function called {\sf Submodular Prize-Collecting Steiner Forest
(\spcsf)}, in which the penalty function $\pi$ is a monotone
non-negative submodular function\footnote{A function
$f:2^{S}\mapsto\R$ is called \emph{submodular} if and only if
$\forall A, B\subseteq S:f(A)+f(B) \geq f(A\cup B)+f(A\cap B)$. An
equivalent characterization is that the marginal profit of each item
should be non-increasing, i.e., $f(A\cup\{a\})-f(A) \leq
f(B\cup\{a\})-f(B)$ if $B\subseteq A\subseteq S$ and $a\in
S\setminus B$. A function $f:2^S\mapsto\R$ is \emph{monotone} if and
only if $f(A) \leq f(B)$ for $A\subseteq B\subseteq S$. Since the
number of sets is exponential, we assume a value oracle access to
the submodular function; i.e., for a given set $T$, an algorithm can
query an oracle to find its value $f(T)$.} of all unsatisfied pairs.
In \pcsf when all penalties are $\infty$, the problem is the classic
APX-hard {\sf Steiner Forest} problem, for which the best known
approximation ratio is $2-\frac{2}{n}$ ($n$ is the number of nodes
of the graph) due to Agrawal, Klein, and Ravi~\cite{AKR95} (see also
\cite{GW95} for a more general result and a simpler analysis). The
case of {\sf Prize-Collecting Steiner Forest} problem in which all
sinks are identical is the classic {\sf (rooted) Prize-Collecting
Steiner Tree (\pcst)} problem. In the unrooted version of this
problem, there is no specific sink (root) and the goal is to find a
tree connecting some sources and pay the penalty for the rest of
them. We also study two variants of {\sf (unrooted) Prize-Collecting
Steiner Tree}, {\sf Prize-collecting
  TSP (\pctsp)} and {\sf Prize-collecting
  Stroll (\pcs)}, in which the set of edges should form a cycle and a path (in order) instead of a tree.
When in addition all penalties are $\infty$ in these
prize-collecting problems, we have classic APX-hard problems {\sf
Steiner Tree}, {\sf TSP} and {\sf Stroll (Path TSP)} for which the
best approximation factors in order are 1.38~\cite{BGRS10},
$\frac{3}{2}$~\cite{Christofides76-TSP}, and
$\frac{3}{2}$~\cite{Hoogeveen91-TSP}.

In network design, planarity is a natural restriction since in
practical scenarios of physical networking, with cable or fiber
embedded in the ground, crossings are rare or nonexistent. Thus
obtaining algorithms with better approximation factors are highly
desirable in this case.  In many cases, approximation algorithms for
planar graphs is based on reducing the problem to bounded treewidth
instances such that the optimum changes only by a small term. This
idea goes back to the classical work of Baker~\cite{cr:3} and have
been applied successfully several times in various contexts. The
algorithmic and graph-theoretic properties of treewidth are
intensively studied and a well-understood dynamic programming
technique can solve NP-hard problems on bounded treewidth graphs. Our
goal is to understand how far this paradigm can be pushed: what are
the most general problems that can be solved this way. In particular,
we want to understand the applicability of this technique to
prize-collecting variants of standard optimization problems.

{\sf TSP}, {\sf Steiner Tree}, and {\sf
Steiner Forest} all have been considered extensively on planar
graphs. Indeed all these problems remain hard even on planar
graphs~\cite{GJ77}. However obtaining a PTAS for each of these
problems remained a very important open problem for several years.
Grigni, Koutsoupias, and Papadimitriou~\cite{papa} obtained the
first PTAS for {\sf TSP} on unweighted planar graphs in 1995 which later
has been generalized to weighted planar graphs \cite{AGK98} (and
improved to linear time \cite{Klein08}). Obtaining a PTAS for
{\sf Steiner Tree} on planar graphs remained elusive for almost 12 years
until 2007 when Borradaile, Klein and Mathieu~\cite{BKM07}
obtained the first PTAS for {\sf Steiner Tree} on planar graphs using a
revolutionary technique of contraction decomposition  and building
spanners and posed obtaining a PTAS for {\sf Steiner Forest} in planar
graphs as the main open problem. Bateni, Hajiaghayi and
Marx~\cite{BHM10} very recently solved this open problem using a new
primal-dual technique for building spanners and obtaining PTASs by
reducing the problem to bounded treewidth graphs. Note that the {\sf Steiner
Forest} problem already shows signs of the reduction to bounded
treewidth paradigm breaking down: surprisingly,
{\sf Steiner Forest} turns out to be NP-hard even on graphs of treewidth
3. However, \cite{BHM10} gets around this problem by using a PTAS on
bounded treewidth graph instead of an exact algorithm.

Obtaining PTASs for prize-collecting versions of these problems
remained a main open problem (see~\cite{BHM10,BH10}). It is not
obvious how to generalize the reduction to bounded treewidth for these
problems, and in particular new techniques are needed for handling
penalties before building a spanner. In this paper, we resolve these
open problems for all three of \pcst, \pctsp, \pcsf, and even
more generally, for \spcsf, by reducing these problems on planar
graphs to the same problems on graphs
of bounded treewidth.  More precisely we show any
$\alpha$-approximation algorithm for these problems on graphs of
bounded treewidth gives a $(\alpha+\epsilon)$-approximation algorithm
for these problems on planar graphs and bounded-genus graphs, for any
constant $\epsilon> 0$. Therefore, we demonstrate that the technique
of reduction to bounded treewidth works even for very general version
of problems involving prizes.  Since \pcst and \pctsp can be solved
exactly on graphs of bounded treewidth using standard dynamic
programming techniques (as we discuss later in the paper), we
immediately obtain PTASs for \pcst and \pctsp on planar graphs (the
same holds for \pcs as well).  In contrast, we show that \pcsf is
APX-hard already on series-parallel graphs, which are planar graphs
of treewidth at most 2, ruling out any hope for a PTAS for planar
\pcsf. This result is interesting on its own, since it gives the first
provable hardness separation between prize-collecting and
non-prize-collecting (regular) versions of the problems: regular {\sf
Steiner Forest} is known to be polynomially solvable on
series-parallel graphs and admits a PTAS on graphs of bounded
treewidth. since {\sf Steiner Forest} on series-parallel graphs is
polynomially solvable and more generally on graphs of bounded
treewidth admits a PTAS~\cite{BHM10}. An analogous hardness result can
be given for Euclidean \pcsf when the vertices of the input graph are
points in the Euclidean plane and the lengths are Euclidean distances
(which answers an open problem in~\cite{BH10}). This ends the common
belief that prize-collecting variants should not add any new hardness
to the problems.

\textbf{Related work.}
\pcst and \pctsp are two of the classic optimization problems with a
large impact, both in theory and practice.  At AT\&T, \pcst code has
been used in large-scale studies in access network design, both as
described in Johnson, Minkoff and Phillips~\cite{JohnsonMP00-pcst},
and another unpublished applied work by Archer at
al. %~\cite{Archerprivatecommunication}.
 The impact of \pctsp within
approximation algorithms is also far-reaching. In particular \pctsp
 is a Lagrangian relaxation of the $k$-MST problem, which asks for
the minimum-cost tree spanning at least $k$ nodes, and has used in a
sequence of papers
(\cite{Garg96-3approx-kMST,AR98-kMST,ChudakRW04-kMST-lr-jrnl,AroraK06-kMST-jrnl})
culminating in a 2-approximation algorithm for $k$-MST by Garg
\cite{Garg05-2approx-kmst}. \pctsp has also been used to improve the
approximation ratio and running time of algorithms for the {\sf
  Minimum Latency} problem (\cite{ArcherLW08-mlp,
ChaudhuriGRT03-mlp-focs}). The first approximation algorithms for
the \pcst and \pctsp problems were given by Bienstock et
al.~\cite{BienstockGSW93-pctsp}, although the \pctsp had been
introduced earlier by Balas~\cite{Balas89-pctsp}. Bienstock et al.\ 
achieved a factor of 3 for \pcst and 2.5 for \pctsp by rounding the
optimal solution to a linear programming (LP) relaxation. Later,
Goemans and Williamson~\cite{GW95-moats} constructed primal-dual
algorithms using the same LP relaxation to obtain a 2-approximation
for both problems, building on work of Agrawal, Klein and
Ravi~\cite{AKR95}.  Chaudhuri et al.\ modified the
Goemans-Williamson algorithm to achieve a 2-approximation algorithm
for \pcs~\cite{ChaudhuriGRT03-mlp-focs}. Improving over the
approximation factor 2 of  Goemans and Williamson for \pcst and
\pctsp was a long-standing open problem for 17 years until recently
that Archer, Bateni, Hajiaghayi, and Karloff~\cite{ABHK09} obtain
constant factors strictly better than 2 ($\approx 1.99$) for both
problems, and for \pcs as well. More recently Goemans combined some
ideas of~\cite{ABHK09} with others from~\cite{Goemans98-PCTSP-talk}
to improve the ratio for \pctsp below
1.915~\cite{Goemans09-PCTSP-impr}.

 The general form of the {\sf Prize-Collecting Steiner Forest} problem first
has been formulated  by Hajiaghayi and Jain~\cite{HJ06}. They showed
how by using a primal-dual method to a novel integer programming
formulation of the problem with doubly-exponential variables, we can
obtain a 3-approximation algorithm for the problem. In addition,
they show that the factor 3 in the analysis of their algorithm is
tight. However they show how a direct randomized LP-rounding
algorithm with approximation factor 2.54 can be obtained for this
problem. Their approach has been generalized by Sharma, Swamy, and
Williamson~\cite{SSW07} for network design problems where violated
arbitrary 0-1 connectivity constraints are allowed in exchange for a
more general penalty function.  Hajiaghayi and Nasri~\cite{HN08}
show factor 3 for {\sf Prize-Collecting Steiner Forest} can  also be
obtained via an iterative rounding approach, first introduced by
Jain~\cite{Jain01}, and indeed factor 3 is the best one can hope via
this approach. The work of Hajiaghayi and Jain has also motivated a
game-theoretic version of the problem considered by Gupta et
al.~\cite{GKL+07}. Very recently, Hajiaghayi et al.~\cite{HKKN10}
obtain a 2.54 approximation algorithm for the more general problem
\spcsf. Aforementioned, our reduction from planar graphs to graphs of
bounded treewidth works even for \spcsf.  It is worth mentioning
optimizing a submodular function, a discrete analog of a convex
function, which also demonstrates economy of scale is a central and
very general problem in combinatorial optimization and has been
subject of a thorough study in the literature in many important
settings including cuts in graphs~\cite{IFF00,GW95,Q95}, plant
location problems~\cite{CFN76, CFN77}, rank function of
matroids~\cite{E70}, set covering problems~\cite{F98}, and certain
restricted satisfiability problems~\cite{H01, FG95}.

\paragraph{Remark}
Subsequent to, and independent of, our work, Chekuri et al.~\cite{CEK10:prize}
obtain a subset of our results including a reduction for prize-collecting 
Steiner tree and prize-collecting Steiner forest from planar graphs
to graphs of bounded treewidth (i.e., a weaker version of our Theorem~\ref{thm:reduce-submod},
albeit with different techniques) which leads to a PTAS for planar prize-collecting
Steiner tree.  The hardness results though are unique to our work.

%%%%%%%%%%%%%%%%%
%%%%%%%%%%%%%%%%%
%%%%%%%%%%%%%%%%

%\subsection{Contributions}
\section{Contributions}
We first formally define the most general problem studied in this paper.
%
%
%\begin{definition}[\spcsf]
An instance of \prob{Submodular Prize-Collecting Steiner Forest} \spcsf
is described by a triple $(G, \D, \pi)$ where
$G$ is a undirected weighted graph,
$\D$ is a set of $d_i = \{s_i, t_i\}$ demand pairs,
and $\pi : 2^\D \mapsto \R^+$ is a monotone nonnegative submodular penalty function.
A demand $d = \{s,t\}$ is \emph{satisfied} by a subgraph $F$
if and only if $s, t$ are connected in $F$.
If a forest $F$ satisfies a subset $\dsat$ of the demands,
its cost is defined as $\cost(F) := \len(F) + \pi(\dunsat)$,
where $\len(F)$ is a shorthand for the total length of all edges in $F$,
and $\dunsat := \D\setminus\dsat$ denotes the subset of unsatisfied demands.

We similarly define \spctsp, \spcs and \prob{SPCST}
that are submodular prize-collecting variants of
\prob{Travelling Salesman Problem}, \prob{Stroll} and \prob{Steiner Tree}, respectively.
The instance is represented by $(G, \D, \pi)$
where all the demands $d=\{s,t\}\in\D$ share
a common root vertex $r \in V(G)$.\footnote{The problems
may be more naturally defined with single-vertex demands rather  demand pairs; 
having such a formulation, we can guess one vertex of the solution,
designate it as the root and obtain the rooted formulation as defined in this paper.}
A solution $F$ is a TSP (stroll or Steiner tree, respectively) for a subset of
demands, say $\dsat \subseteq\D$.
The cost is then $\cost(F) := \len(F) + \pi(\dunsat)$, where
$\dunsat := \D\setminus\dsat$.
%\end{definition}

We first show that \psubmodforest{} on planar graphs (or more generally, bounded-genus graphs)
is almost equivalent to that on graphs of bounded-treewidth;
refer to Appendix~\ref{app:btw-defs} for definitions regarding 
the treewidth and bounded-treewidth graphs as well as bounded-genus graphs.
In particular, were we able to give 
a PTAS for \spcsf
on graphs of bounded treewidth, we would readily have a PTAS for \spcsf on bounded-genus graphs.
In the rest of the paper, we focus on planar graphs.
All the algorithms and analyses can be extended with minor modifications
to work for bounded-genus graphs.

\begin{theorem}\label{thm:reduce-submod}
For any given constant $\eps > 0$,
 an $\alpha$-approximation algorithm for \spcsf on graphs of bounded treewidth
gives a $(\alpha+\eps)$-approximation algorithm for \spcsf on planar graphs.
\end{theorem}

%\xxx{the above theorem is true, but if we want to be very precise---in fact, for \pcst we kind of need this---we should say that the structure of the penalty function does not change significantly.  something like $\pi^{new}(D) := \pi^{old}(D \cup \D_1 \setminus \D_2) \:\forall D$; otherwise,
%it is theoretically possible that we start with \pcst and get \spcsf.}

%\xxx{where's the best place for the following paragraph?  needs edits.}
The core of the reduction is based on a \emph{prize-collecting clustering} technique
that was first implicitly used in \cite{ABHK09} and later developed in \cite{BHM10}.
In this work, the clustering technique is generalized as follows:
First, we need to extend the ideas to work for prize-collecting variants
of Steiner network problems.  This can indeed make the problem provably harder; see Theorem~\ref{thm:sf-hard}.
The original prize-collecting clustering associates a potential value to each
node and grows the corresponding clusters consuming these potentials.
However, in order to extend it to the prize-collecting setting, we
consider source-sink potentials.  This means that there is some interaction between the
potentials of different nodes.
Secondly, we consider submodular penalty functions that
model even more interaction between the demands.
The extended prize-collecting clustering procedure has two phases.
In the first phase, we have a source-sink moat-growing algorithm,
and in the second phase, we have a single-node potential moat-growing
like \cite{BHM10}.

Section~\ref{sec:reduce} is devoted to the formal proof Theorem~\ref{thm:reduce-submod}.
The algorithm starts with a constant-approximate solution $F^1$, say, 
obtained using Hajiaghayi et al.~\cite{HKKN10} who prove a $3$-approximation
for \spcsf on general graphs.
The forest $F^1$ satisfies a subset of demands,
and we know the total penalty of unsatisfied demands is bounded.
The algorithm then tries to satisfy more demands by constructing
a forest $F^2 \supseteq F^1$ whose length is bounded; 
see \algo{RestrictDemands} in Section~\ref{sec:rest-prize}.
This step heavily uses a \emph{submodular prize-collecting clustering}
algorithm\footnote{The algorithm bears some similarity to the primal-dual moat-growing
algorithms for the Steiner network problems.  One key difference is that we do not have
a primal LP.  We have an LP similar to the dual linear programs used in such algorithms,
and we use a notion of potential as a substitute for the lack of the primal LP.  The potentials,
among other things, play the role of an upper bound for the value of the dual LP.} 
introduced in Section~\ref{sec:submod-cluster}.
At the end of this step, we can assume that the near-optimal
solution does not satisfy the demands which are unsatisfied in $F^2$.
Submodularity poses several difficulties in proving this property:
ideally, we want to say that the cost paid by the optimal solution
to satisfy these demands is significantly more than their penalty
value.  Surprisingly, this is not true.  Nevertheless, we can prove
that the \emph{marginal cost} of the demands satisfied in the near-optimal
solution but not in $F^2$ can be charged to the cost the near-optimal
solution pays in order to satisfy them.
The next step of the reduction is to build a forest $F^3 \supseteq F^2$
of bounded length that may connect several components of $F^2$ together;
see Section~\ref{sec:reduce:2}.
This is done by assigning to each component of $F^2$ a potential 
proportional to its length, and then running a prize-collecting clustering
similar to that of \cite{BHM10}.
This guarantees that the near-optimal solution does not need to connect
different components of $F^3$ to each other.  The implication is that
we can construct a spanner (see \cite{BHM10, BKM07, Klein08})
out of each component of $F^3$ separately
from the others.  In the previous work \cite{BHM10}, we could solve
each of the subinstances independently, however, the penalty interaction
originating from the submodular penalty function in the current work
does not allow us to solve each subinstance completely independently.
Instead, we say that the forest of the near-optimal solution on each
subinstance is independent of the others.
After constructing the spanner graph $F^4$, we invoke a generalization
of the shifting idea of Baker~\cite{cr:3} due to \cite{DHM07,Klein08}.
Paying a cost of at most $\epsilon\,\OPT$, we end up with
a graph of bounded treewidth.

%We prove the above theorem in Section~\ref{sec:reduce}.
Since bounded-treewidth graphs bear some similarity to trees,
several tools have been developed for solving optimization
problems on them.
Standard techniques, see Appendix~\ref{app:btw}, allow us
to obtain PTASs for several Steiner network problems on 
graphs of bounded treewidth.

\begin{theorem}\label{thm:btw}
 \pcst, \pcs and \pctsp admit PTASs on bounded-treewidth graphs.
\end{theorem}

In Section~\ref{sec:planar} we show how this results in PTASs for
the above problems on planar graphs.
In particular, this is simple for \pcst since it is a special case of \spcsf.
For the other two problems, however, refer to the discussion in Section~\ref{sec:planar}.

In contrast, we show \textsf{Prize-Collecting Steiner Forest} is
APX-hard, even on planar graphs of treewidth at least two;
Hajiaghayi and Jain show the problem can be solved in polynomial
on tree metrics~\cite{HJ06}.

\begin{theorem}\label{thm:sf-hard}
 \pcsf is APX-hard on (1) planar graphs of treewidth two and on (2) the two-dimensional Euclidean metric.
\end{theorem}

This is done via a reduction from \textsf{Bounded-Degree Vertex Cover}
in Appendix~\ref{sec:hard}.  Indeed, the result shows that
\textsf{Submodular Prize-Collecting Steiner Tree} (the version of the
problem when the solution has to be a connected tree instead of a
forest) is also APX-hard.  This implies the hardness of \pcsf
originates from the interaction between the penalties of terminals
rather than from the different components of the solution.

Surprisingly, the hardness also works
for Euclidean metrics, answering an open question raised in \cite{BH10}.  
This is a very rare instance where a natural network optimization problem
is APX-hard on the two-dimensional Euclidean plane. % We do not know of any other instance of a similar hardness.

Theorem~\ref{thm:sf-hard} means
that planar \pcsf\ reaches a level of complexity where even though
reduction to bounded treewidth instances works, it does not give us
a PTAS for the problem (in fact, no PTAS exists unless
$\textup{P}=\textup{NP}$). However, the treewidth reduction approach
can be still useful for obtaining constant factor approximations for
planar graphs better than the factor 2.54 algorithm of \cite{HJ06} for
general graphs. Theorem~\ref{thm:reduce-submod} show that beating the 2.54 factor on
bounded treewidth graphs would immediately imply the same for planar
graphs. We pose it as an open question whether this is indeed possible
for \pcsf.

%\xxx{TODO: do you mention the \emph{Vertex Cover} solution for \pcsf on trees?  I didn't see it.}

\section{Reduction to bounded-treewidth case}\label{sec:reduce}

This section focuses on proving Theorem~\ref{thm:reduce-submod}.
In fact, we prove a stronger version of the theorem, that is necessary for
obtaining PTASs for \pcst, \pctsp, and \pcs.
We reduce an instance $(G, \D, \pi)$ of \spcsf to an instance
$(H, \D, \pi')$ where $H$ has bounded treewidth and
$\pi'$ has a structure similar to $\pi$; in particular, for some $\dunsat \subseteq \D$
%and a positive parameter $c$, 
we define
$\pi'(D) :=  \pi(D \cup \dunsat)$ for all $D \subseteq \D$.
Notice that if $\pi$ is submodular, then so is $\pi'$.
Moreover, if $\pi$ models a \pcsf instance, 
i.e., $\pi$ is an additive function,
then $\pi'(D) - \pi'(\emptyset)$ models a \pcsf instance, too.  
In fact, $\pi'(D)$ is an additive function that is shifted
with a fixed amount $\pi'(\emptyset)$.
Same condition holds for \pcst, \pctsp and \pcs.
Therefore, after reducing a \pcst instance, we are left with a
\pcst instance---rather than an \spcsf one---on a bounded-treewidth graph.

The proof has three steps:
\begin{enumerate}
 \item We start with an instance $(G, \D, \pi)$ of \spcsf.
We first take out a subset, say $\dunsat$, of demands 
whose cost of satisfying is too much compared to
their penalties.  Thus, we can focus on the remaining demands, 
say $\dsat := \D\setminus\dunsat$.
 \item Afterwards, we partition the remaining demands $\dsat$ into $\D_1, \D_2, \dots, \D_p$
such that, roughly speaking, \spcsf can be solved separately on each of the demand sets without
increasing the total cost substantially.
 \item Finally, we build a \emph{spanner} for each demand set $\D_i$,
and use similar ideas as in \cite{BHM10} to reduce the problem to bounded-treewidth
graphs.
\end{enumerate}

The first step is carried out in the following theorem.
The proof appears in Section~\ref{sec:reduce:1}, 
and uses a submodular prize-collecting clustering technique introduced
in Section~\ref{sec:submod-cluster}.
This step allows us to focus on only a subset $\dsat$ of demands, and ignore the
rest of the demands.  The additional cost due to this is only $\eps\OPT$.

\begin{theorem}\label{thm:reduce}
  Given an instance $(G, \D, \pi)$ of \spcsf (or \spctsp or \spcs)
and a parameter $\eps > 0$,
we can construct in polynomial time a subgraph $F$ of $G$, satisfying only a subset 
$\dsat\subseteq\D$ of demands, in effect leaving $\dunsat := \D \setminus \dsat$ unsatisfied,
such that
\begin{enumerate}
 \item $\len(F) \leq (6\eps^{-1} + 3) \OPT$, and %\xxx{check this}
% \item $\len(F') \geq \eps^{-1} \cdot \pi(\D'')$ if subgraph $F'$ satsfies 
%$\D'' \subseteq \D \setminus \D'$.
 \item the optimum of $(G, \dsat, \pi')$ is at most $(1+\eps)\OPT$
where $\pi'(D) := \pi(D \cup \dunsat)$ is defined for $D\subseteq \dsat$.
%there exists a solution $\OPT'$  whose cost does not exceed $(1+\eps) \OPT$,
%and $\OPT'$ pays the penalty for a superset of $\D \setminus \D'$.
\end{enumerate}
\end{theorem}

%\xxx{i'd like to say the above theorem works for \pctsp and \pcs too.  this would make the later
%conclusions more precise and immediate.}

At this point, we have a constant-approximate solution satisfying all the (remaining)
demands.  
The second step is a generalization and extension of the work in \cite{BHM10}.
We are trying to break the instance into smaller pieces.  The solution to
each piece is almost independent of the others, i.e., there is little interaction between
them.  The following theorem is proved in Section~\ref{sec:reduce:2}.

\begin{theorem}\label{thm:break}
 Given are an instance $(G, \D, \pi)$ of \spcsf,
a forest $F$ satisfying all the demands,
and a parameter $\eps > 0$.
We can compute in polynomial time a set of trees $\{\hat T_1,
\dots, \hat T_k\}$, and a partition of demands $\{\D_1, \dots, \D_k\}$,
 with the following properties.
\begin{enumerate}
 \item
   All the demands are covered, i.e., $\D=\bigcup_{i=1}^k \D_i$.
 \item
   The tree $\hat T_i$ spans all the terminals in $\D_i$.
 \item
  The total length of the trees $\hat T_i$ is within a constant factor of the length of $F$, i.e., $\sum_{i=1}^k \len(\hat T_i) \leq (\frac{2}{\eps}+1)\len(F)$. %\xxx{correct constants?}
 \item
  Let $\D^\ast$ be the subset of demands satisfied by $\OPT$.
  Define $\D^\ast_i := \D^\ast \cap \D_i$, and denote by $\sfor(G, \D)$ the
length of a minimum Steiner forest of $G$ satisfying the demands $\D$.
%   The sum of the costs of  minimum Steiner forests of all demand sets $\D_i$ is no more than %$1+\eps$ times the cost of a minimum Steiner forest of $G$; 
%The sum of the optimal solutions for each demand set $\D_i$ is close
%to the optimal solution of $\D$; in particular,
%$\sum_i\OPT_{\D_i}(G) \leq \OPT_{\D}(G) + \eps\cdot\len(F)$. %, where $\OPT_i$ denotes the cost of a minimum Steiner forest for $\DD_i$.
We have $\sum_i\sfor(G, \D^\ast_i) \leq (1+\eps) \sfor(G, \D^\ast)$.
%$\sum_i\sfor(G, \D^\ast_i) \leq \sfor(G, \D^\ast) + \eps\cdot\len(F)$.
%\xxx{need to only consider the length of the solution here, not the penalties!}
\end{enumerate}
\end{theorem}

The final step is very similar to the spanner construction of \cite{BHM10,BKM07}.
Since it has been extensively covered in those works, we defer
the details to the full version of the paper.

%\xxx{say where each theorem is proved}
%\xxx{show why the above results imply the promised theorem.}

%Now given the above theorems we can prove the main theorem.
Now we show how the above theorems imply the main theorem of the paper.
\begin{proof}[\proofname\ of Theorem~\ref{thm:reduce-submod}]
 Start with an instance $(G, \D, \pi)$ of \spcsf.
Without loss of generality we present an approximation guarantee of $\alpha + O(1) \eps$.
Find $F$, $\dsat$ and $\dunsat$ from applying Theorem~\ref{thm:reduce} on $(G, \D, \pi)$.
We know that $F$ satisfies $\dsat$ and $\len(F) = O(\OPT)$.
Moreover, $\OPT_{\dsat}(G) \leq \OPT$.
Define $\pi^+(D) := \pi(D \cup \dunsat)$ for all $D \subseteq \D$.
Clearly the optimal solution of $(G, \dsat, \pi^+)$ costs no more than $(1+\eps)\OPT$.
Pick $\eps' < \eps\cdot \len(F)/\OPT$ and 
feed $(G, \dsat, \pi^+)$ along with $F$ and $\eps'$ to Theorem~\ref{thm:break}, 
in order to obtain $\D_i$'s and $\hat T_i$'s for $i=1, \ldots, k$.
We have $\sum_i\len(\hat T_i) = O(\len(F)) = O(\OPT)$ since $\eps'$ is a constant.
In addition, the theorem guarantees a near-optimal solution $\OPT^+$
of cost at most $(1+2\eps)\OPT$ that does not use the
connectivitiy of different components $\D_i$ and $\D_{i'}$ for $i,i'\in\{1,\ldots,k\}:i\neq i'$. 
This ensures that the spanner construction gives us a graph $G^+$ (of total length $O(\OPT)$)
that approximate the forest of the solution within a $1+\eps$ factor.
Thus, the optimal solution of $(G^+, \dsat, \pi^+)$ costs
at most $(1+\eps)(1+2\eps) \OPT = [1+O(1)\eps] \OPT$.
Since the total length of the graph $G^+$ is within $O(\OPT)$,
we can use the decomposition theorem of \cite{DHM07}
to reduce the problem to bounded-treewidth graphs with an increase
of $\eps\OPT$ in the solution cost.
The reduced instance is solved via the $\alpha$-approximation algorithm,
and we finally get an approximation ratio of $\alpha+O(\eps)$.
%We solve each demand set separately and finally pay for the additional penalty
%of $\D \setminus \D'$.
%The cost is at most $\OPT_{\D'}(G) + \eps \len(F) \leq \OPT + \eps \len(F) \leq (1+\eps') \len(F)$
%if $\eps$ is picked sufficiently large, i.e.,  $\eps \leq \eps' \len(F)/\OPT$.
%
%Next we can use a theorem due to \cite{DHM07} to buy a small fraction of
%edges, i.e., at most $\eps\OPT$, that leaves us with a bounded-treewidth graph.
%\xxx{need more details here}
%\xxx{I think we get $\alpha+\eps$ rather than $(1+\eps)\alpha$.
%we also need to say we start with a three appx.}
\end{proof}

%\xxx{do we want to say anything about higher connectivities?}

% In this section, we prove Theorem~\ref{thm:reduce} and Theorem~\ref{thm:break} that
%allow us to reduce several network design problems on
%planar graphs to graphs of bounded treewidth.
%
%DEF: We write $D \odot S$ if and only if $0 < D \cap S < |D|$.
%This means that the demand $D$ is cut by the set $S$.
%We then say that $D$ is a \emph{relevant} demand of $S$,
%or $S$ is a relevant cut for $D$.

%\subsection{Submodular growth algorithm}\label{sec:submod-cluster}
\subsection{Submodular prize-collecting clustering}\label{sec:submod-cluster}

First we present and analyze a primal-dual algorithm for 
\spcsf, and later 
we see how this algorithm can be used to achieve the goal of
identifying and removing certain demands from the optimal solution
such that the additional penalty is negligible.

Consider an instance $(G(V,E), \D, \pi)$ of the \spcsf.
A set $S \subseteq V$ is said to \emph{cut} a demand $d=\{s,t\}$
if and only if $|S \cap d| = 1$.  We denote this by the short-hand $d \odot S$,
and say the demand $d$ \emph{crosses} the set $S$.
In the linear program~\eqref{eqn:s-lp:1}--\eqref{eqn:s-lp:3},
there is a variable $y_{S,d}$ for any $S\subseteq V$, $d\in\D$ such
that $d \odot S$.
Conveniently, we use the short-hands $y_S := \sum_{d\in\D} y_{S,d}$
and $y_d := \sum_{S\subseteq V} y_{S,d}$.

\begin{lp}
& \sum_{S: e\in\delta(S)}\: y_S \leq c_e &\hspace{3cm}&&
       \forall e\in E \label{eqn:s-lp:1}\\
& \sum_{d\in D} y_d \leq \pi(D) &&&\forall D\subseteq \D  \label{eqn:s-lp:2}\\
& y_{S,d} \geq 0 &&&\forall d\in\D, S\subseteq V, d\odot S. \label{eqn:s-lp:3}
\end{lp}

%\xxx{BUUUUQ}
 We produce a solution to the above LP.
Theorem~\ref{thm:reduce} is proved via some properties of this solution.
 These constraints look like the dual of a natural linear program for 
\spcsf.
For the sake of convenience, we use the notation $y(D) := \sum_{d\in D} y_d$ for any $D \subseteq \D$.
%Interestingly, however, Hajiaghayi et al.~\cite{HKKN10} show that this 
%natural LP unbounded integrality gap.  Yet, we give a primal-dual algorithm
%below, and prove some guarantees about it that will prove useful in the
%main reduction procedure.

\begin{lemma}\label{lem:submod-growth}
Given an instance $(G, \D, \pi)$ of \spcsf,
we produce in polynomial time a forest $F$ and a subset $\dunsat \subseteq \D$ of demands, 
along with a feasible vector $y$ for the above LP such that
\begin{enumerate}
 \item $y(\dunsat) = \pi(\dunsat)$;
 \item $F$ satisfies any demand in $\dsat := \D \setminus \dunsat$;
and
 \item $\len(F) \leq 2 y(\D)$.
% \item
\end{enumerate}
\end{lemma}

The solution is built up in two stages.
First we perform an \emph{submodular growth} to find a forest $F_1$ and a corresponding $y$ vector.
This is different from the usual growth phase of \cite{GW95, AKR91} in that
the penalty function may go tight for a set of vertices that are not currently connected.
In the second stage, we \emph{prune} some  edges of $F_1$ to obtain another forest $F_2$.
%Each tree component of the forest specifies a set of the partition $\PP$ and  replacing each contracted vertex by the original tree
%Uncontracting the trees $T^\ast_i$  turns $F_2$ into the Steiner trees $T_i$ in the statement of Theorem~\ref{thm:break}.
Below we describe the two phases of Algorithm~\ref{alg:submod-pc-cluster} (\algo{Submodular-PC-Clustering}).

\paragraph{Growth}
 We begin with a zero vector $y$, and an empty set $F_1$.
 A demand $d\in\D$ is said to be \emph{live} if and only if
$x(D) < \pi(D)$ for any $D \subseteq \D$ that $d\in D$.  
If a demand is not live, it is \emph{dead}.
 During the execution of the algorithm \algo{Submodular-PC-Clustering}, 
we maintain a partition $\C$ of vertices $V$ into clusters; 
it initially consists of singleton sets.
 Each cluster is either \emph{active} or \emph{inactive};
  the cluster $C\in\C$ is \emph{active} if and only if there is a live demand $d: d\odot C$.
We simultaneously \emph{grow} all the active clusters by $\eta$.
 In particular, if there are $\kappa(C)>0$ live demands crossing an active cluster $C$, 
we increase $y_{C,d}$ by $\eta/\kappa(C)$ for each live demand $d: d\odot C$.
 Hence, $y_{C}$ is increased by $\eta$ for every active cluster $C$.
 We pick the largest value for $\eta$ that does not violate any of the constraints in \eqref{eqn:s-lp:1} or \eqref{eqn:s-lp:2}.
 Obviously, $\eta$ is finite in each iteration because the 
values of these variables cannot be larger than $\pi(\D)$.
 Hence, at least one such constraint goes tight after each growth step.
 If this happens for an edge constraint for $e=(u,v)$,
 then there are two clusters $C_u\ni u$ and $C_v\ni v$ in $\C$,
 at least one of which is growing.
 We merge the two clusters into $C = C_u \cup C_v$ by adding the 
edge $e$ to $F_1$, remove the old clusters and add the new one to $\C$.
 Nothing needs to be done if a constraint \eqref{eqn:s-lp:2} becomes tight.
The number of iterations is at most $2|V|$ because at each 
event either a demand dies, or the size of $\C$ decreases.

%In the previous work~\cite{BHM10}, 
Computing $\eta$ is nontrivial here.
In particular, we have to solve an auxiliary linear program to find its value.
New variables $y^\ast_{S,d}$ denote the value of vector $y$ after a growth of
size $\eta$.  All the constraints are written for the new variables.
There are exponentially many constraints in this LP, however,
it admits a separation oracle and thus can be optimized.\footnote{
Notice that there are only a polynomial number of non-zero variables at each step
since $y_{S,d}$ may be non-zero only for clusters $S$, and these clusters form a
laminar family in our algorithm.
Verifying constraints \eqref{eqn:eta:1}-\eqref{eqn:eta:3} and \eqref{eqn:eta:5} is very simple.
Verifying constraints \eqref{eqn:eta:4} is equivalent to finding 
$\min_{D\subseteq \D} \pi(D) - y^\ast(D)$ and checking that it is non-negative.
The function to minimize is submodular and thus can be minimized in polynomial time~\cite{IFF00}.
A standard argument shows that the values of these variables have polynomial size.
We defer to the full version of the paper the detailed discussion of
how the LP can be approximated. 
} %\xxx{more info}
\begin{lp}
\maximize &\eta \label{eqn:eta:0}\\
\subject& y^\ast_{S,d} = y_{S,d} + \frac{\eta}{\kappa(S)} &&&\forall d\in \D, S\subseteq V, d\odot S, \kappa(S) > 0 \label{eqn:eta:1}\\
& y^\ast_{S,d} = y_{S,d} &&&\forall d\in \D, S\subseteq V, d\odot S, \kappa(S) = 0 \label{eqn:eta:2}\\
&\sum_{S: e\in\delta(S)}\: y^\ast_S \leq c_e &\hspace{0cm}&&
       \forall e\in E \label{eqn:eta:3}\\
& \sum_{d\in D} y^\ast_d \leq \pi(D) &&&\forall D\subseteq \D  \label{eqn:eta:4}\\
& y^\ast_{S,d} \geq 0 &&&\forall d\in\D, S\subseteq V, d\odot S. \label{eqn:eta:5}
\end{lp}

\xxx{maybe one intuitive par on the growth process}
\xxx{TODO:mention at some point why the submodular algorithm is so strange.  many simpler methods fail!}

\paragraph{Pruning}
Let $\eS$ denote the set of all clusters formed during the execution
of the growth step.  It can be easily observed that the clusters $\eS$
are laminar and the maximal clusters are the clusters of $\C$. 
In addition, notice that $F_1[C]$ is connected for each $C\in\eS$.
 
 Let $\B\subseteq\eS$ be the set of all clusters $C$ that do not
cut any live demand.  
Notice that a demand $d$ may still be live at the end of the growth stage
if it is satisfied; roughly speaking, the demand is satisfied before it exhausts
its \emph{potential}.
%, namely,
%for each demand $d: d \odot C$, we have $\sum_S y_{S,D} = \phi_D$.  
In the pruning stage, we
iteratively remove edges from $F_1$ to obtain $F_2$.  More
specifically, we first initialize $F_2$ with $F_1$.  Then, as long as
there is a cluster $S\in\B$ %and a tree component $T$ of $F_2$
such that $F_2\cap\delta(S) = \{e\}$, we remove the edge $e$ from
$F_2$. 

A cluster $C$ is called a \emph{pruned cluster} if it is pruned in the
second stage in which case, $\delta(C)\cap F_2 = \emptyset$.  Hence, a
pruned cluster cannot have non-empty and proper intersection with a
connected component of $F_2$.

\begin{algorithm}
\caption{\algo{Submodular-PC-Clustering}\label{alg:submod-pc-cluster}}
\textbf{Input:} Instance $(G(V,E), \D, \pi)$ of \prob{Generalized prize-collecting Steiner forest}\\
\textbf{Output:} Forest $F$, subset of demands $\dunsat$ and fractional solution $y$.

\begin{algorithmic}[1]
\STATE Let $F_1\gets\emptyset$.
 \STATE Let $y_{S,d}\gets 0$ for any $d\in\D, S\subseteq V, d\odot S$.
 \STATE Let $\eS\gets\C\gets\left\{ \{v\}: v\in V^\ast\right\}$.
% \STATE Define $\C(v) := C\in\C$ where $v\in C$.
% \STATE Define a vertex $v$ as \emph{live} if and only if $\sum_S y_{S,v} < \phi_v$.
% \STATE Define $\kappa(S)$ as the number of live vertices in $S\in\C$.
 \WHILE {there is a live demand}
%  \STATE Let $\eta$ be the largest possible value such that simultaneously increasing  $y_{\C(v),v}$ by $\frac{\eta}{\kappa(\C(v))}$ for all live vertices $v$ does not violate Constraints~\eqref{eqn:lp:1}-\eqref{eqn:lp:3}.
  \STATE Compute $\eta$ via LP~\eqref{eqn:eta:0}:
the largest possible value such that simultaneously 
increasing  $y_{C}$ by $\eta$ for all active clusters $C \in \C$ does not violate Constraints~\eqref{eqn:s-lp:1}-\eqref{eqn:s-lp:3}.
%  \STATE Let $y_{\C(v),v}\gets y_{\C(v),v}+\frac{\eta}{\kappa(\C(v))}$ for all live vertices  $v$.
  \STATE Let $y_{C,d}\gets y_{C,d}+\frac{\eta}{\kappa(C)}$ for all live demands $d$ crossing clusters $C \in \C$, i.e., $d \odot C$.
  \IF {$\exists e\in E$ that is tight and connects two clusters $C_1$ and $C_2$}
   \STATE Pick one such edge $e=(u,v)$.
   \STATE Let $F_1\gets F_1\cup \{e\}$.
   \STATE Let $C \gets C_1 \cup C_2$.
   \STATE Let $\C\gets \C\cup \{C\} \setminus \{C_1, C_2\}$.
   \STATE Let $\eS \gets \eS \cup \{C\}$.
  \ENDIF
\ENDWHILE
\STATE Let $F_2 \gets F_1$.
\STATE Let $\B$ be the set of all clusters $S\in\eS$ that
do not cut any live demands.
\WHILE {$\exists S\in\B$ such that $F_2\cap\delta(S)=\{e\}$ for an edge $e$}
  \STATE Let $F_2 \gets F_2 \setminus \{e\}$.
\ENDWHILE
\STATE Let $\dunsat$ denote the set of dead demands.
\STATE Output $F := F_2$, $\dunsat$ and $y$.
\end{algorithmic}
\end{algorithm}

%%%%%%%%%%%%%%%%%%%%%

%\xxx{say how each step of the algorithm is carried out.}

We first bound the length of the forest $F$.
The following lemma is similar to the analysis of the algorithm in~\cite{GW95}.
However, we do not have a primal LP to give a bound on the dual.
%In fact, as was pointed out earlier, this LP does not have a bounded integrality
%gap for \prob{Generalized prize-collecting Steiner forest}.
Rather, the upper bound for the length is $\pi(\D)$.
In addition, we bound the cost of a forest $F$ that may have
more than one connected component, whereas the prize-collecting Steiner tree
algorithm of~\cite{GW95} finds a connected graph at the end.
\begin{lemma}\label{lem:submod-growth:cost}
 The cost of $F_2$ is at most $2y(\D)$.
\end{lemma}

\begin{proof}
% The strategy is to prove that the cost of this forest is at most $2\sum_{d,S}y_{S,d} \leq 2\pi(\D)$.
%The equality follows from Equation~\eqref{eqn:p:lp:2}.
%%
Recall that the growth phase has several events corresponding to an edge or set constraint going tight.
We first break apart $y$ variables by epoch.  Let $t_j$ be the
  time at which the $j^{\rm th}$ event point occurs in the growth phase ($0=t_0\leq t_1 \leq t_2 \leq \cdots$), so the $j^{\rm th}$ epoch is the interval of time from $t_{j-1}$ to $t_j$.  For each cluster $C$, let $y_{C}^{(j)}$  be the amount by which $y_C$ grew during epoch $j$, which is  $t_j-t_{j-1}$ if it was active during this epoch, and zero otherwise.  Thus, $y_C = \sum_j y_{C}^{(j)}$.
Because each edge $e$ of $F_2$
  was added at some point by the growth stage when its edge packing constraint \eqref{eqn:s-lp:1} became tight, we can exactly apportion the cost
  $c_e$ amongst the collection of clusters $\{C : e\in \delta(C)\}$ whose
  variables ``pay for'' the edge, and can divide this up further
  by epoch.  In other words, $c_e = \sum_j \sum_{C:e\in \delta(C)}
  y_{C}^{(j)}$.  We will now prove that the total edge cost from $F_2$
  that is apportioned to epoch $j$ is at most $2 \sum_{C} y_{C}^{(j)}$.  In other words, during each epoch,
  the total rate at which edges of $F_2$ are paid for by all active
  clusters is at most twice the number of active clusters.
  Summing over the epochs yields the desired conclusion.

  We now analyze an arbitrary epoch $j$.  Let $\C_j$ denote the set
  of clusters that existed during epoch $j$.
Consider the graph $F_2$, and then collapse each cluster
  $C \in \C_j$ into a supernode.  Call the resulting graph $H$.
%%%%%%%
Although the nodes of
  $H$ are identified with clusters in $\C_j$, we will continue to refer
  to them as clusters, in order to to avoid confusion with the nodes of
  the original graph.  Some of the clusters are active and some may be
  inactive.  Let us denote the active and inactive clusters in $\C_j$
  by $\C_{act}$ and $\C_{dead}$, respectively.
  The edges of $F_2$ that are being partially paid for during epoch $j$
  are exactly those edges of $H$ that are incident to an active cluster,
  and the total amount of these edges that is paid off during epoch
  $j$ is $(t_j-t_{j-1}) \sum_{C \in \C_{act}} \deg_H(C)$.
  Since every active cluster grows by exactly $t_j-t_{j-1}$ in epoch
  $j$, we have $\sum_{C} y_{C}^{(j)} \geq \sum_{C \in\C_j}y_{C}^{(j)} = (t_j-t_{j-1}) |\C_{act}|$. %%\footnote{The inequality may be strict if some $C\in\goodcomp$ is growing in epoch $j$, and hence $\duale{C}{j}>0$, but $C\cap V(F) = \emptyset$, so $C \notin \scc_j$.}
   Thus, it suffices to show that $\sum_{C
    \in \C_{act}} \deg_H(C) \leq 2 |\C_{act}|$.

First we
  must make some simple observations about $H$.  Since $F_2$ is a subset of the edges in $F_1$, and each cluster represents a
  disjoint induced connected subtree of $F_1$, the contraction to $H$ introduces  no cycles.  Thus, $H$ is a forest.
  All the leaves of $H$ must
  be live clusters because otherwise the corresponding cluster $C$ would be
  in $\B$ and  hence would have been pruned away.

  With this information about $H$, it is easy to bound $\sum_{C \in
    \C_{act}} \deg_H(C)$.
  The total degree in $H$ is at most $2(|\C_{act}|+|\C_{dead}|)$.
  Noticing that the degree of dead clusters is at least two,
  we get $\sum_{C\in\C_{act}}\deg_H(C) \leq 2(|\C_{act}|+|\C_{dead}|) - 2|\C_{dead}| = 2|\C_{act}|$ as desired.
\end{proof}

Now we can prove Lemma~\ref{lem:submod-growth}
that characterizes the output of \algo{Submodular-PC-Clustering}.

\begin{proof}[\proofname\ of Lemma~\ref{lem:submod-growth}]
For every demand $d \in \dunsat$ we have a set $D \ni d$
such that $y(D) = \pi(D)$.  The definition of $\dunsat$ guarantees
$D \subseteq \dunsat$.  Therefore, we have sets $D_1, D_2, \dots, D_l$
that are all tight (i.e., $y(D_i) = \pi(D_i)$) and they span $\dunsat$
(i.e., $\dunsat = \cup_i D_i$).
To prove $y(\dunsat) = \pi(\dunsat)$, we use induction
and combine $D_i$'s two at a time.
For any two tight sets $A$ and $B$ we have
$y(A \cup B) = y(A) + y(B) - y(A\cap B) = \pi(A) + \pi(B) - y(A\cap B)
\geq \pi(A) + \pi(B) -\pi(A\cap B) \geq \pi(A \cup B)$,
where the second equation follows from tightness of $A$ and $B$,
the third step is a result of Constraint~\eqref{eqn:s-lp:2},
and the last step follows from submodularity.
Constraint~\eqref{eqn:s-lp:2} has it that $\pi(A\cup B) \geq y(A \cup B)$,
therefore, it has to hold with equality.

Clearly, at the end of execution of \algo{Submodular-PC-Clustering},
any live demand is already satisfied.  Notice that such demands are not
affected in the pruning stage.  Hence, only dead demands may be not satisfied.
This guarantees the second condition.  The third condition follows from
Lemma~\ref{lem:submod-growth:cost}.
\end{proof}

\subsection{Restricting the demands}\label{sec:reduce:1}\label{sec:rest-prize}

We prove Theorem~\ref{thm:reduce} in this section.
First, we obtain a constant-factor approximate solution $F^+$---%
this can be done, e.g., via the $3$-approximation algorithm for general graphs \cite{HKKN10}.
Let $\D^+$ denote the demands satisfied by $F^+$.
 We denote by $T^+_j$ the connected components of $F^+$.
 For each demand $d=\{s,t\}\in\D^+$ we clearly have $\{s,t\} \subseteq V(T^+_j)$
for some $j$.  However, for an unsatisfied demand $d'=\{s',t'\} \in \D \setminus \D^+$,
the vertices $s'$ and $t'$ belong to  two different components of $F^+$.
%Thus, each vertex $v\in V^\ast$ is either one component
%$T^+_j$ or a vertex of $G$ that is not in $F^+$.
%Moreover, each edge $e\in E^\ast$ corresponds to an edge of $G$ that is not inside $F^+$.
%Notice that there may be parallel edges (or self-loops) in $G^\ast$.
% We also update $\D$ and $\pi$ accordingly.
Construct $G^\ast$ from $G$ by reducing the length of edges of $F^+$ to zero.
%In particular, we have a demand $\{s^\ast, t^\ast\} \in \D^\ast$
%if and only if $s^\ast, t^\ast \in V^\ast$ are the super-nodes corresponding \xxx{explain better--also in alg figure}
%to an unsatisfied demand $\{s, t\} \in \D \setminus \D^+$.
 The new penalty function $\pi^\ast$ is defined as follows:
\begin{align}\label{eqn:new-pi}
% \pi^\ast(D) &:= \eps^{-1} \pi(D \cup (\D\setminus\D')) &\text{for } D \subseteq \D.
 \pi^\ast(D) &:= \eps^{-1} \pi(D) &\text{for } D \subseteq \D.
\end{align}
 Finally we run \algo{Submodular-PC-Clustering} on $(G^\ast, \D, \pi^\ast)$; see Algorithm~\ref{alg:prize-restrict}. %,
%and uncontract trees $T^+_j$ in the output.

%we run a variant of \algo{Prize-collecting Clustering} in which
%each super-node receives some potential corresponding to its relevant
%unsatisfied demands.
% Unlike the previous version of \algo{Prize-collecting Clustering}, the current variant
%does not assign fixed potentials to each super-node.  Instead, each unsatisfied demand
%has a fixed potential that can be consummed by all its relevant super-nodes.
%In particular, if more than one such super-nodes grow simultaneously, the rate of
%consumption of the potential will be more than one.

%Before, we can give the algorithm, we need to define $\C(v)$ as the cluster currently containing the vertex $v\in V'$, i.e., $\C(v) := C$ for any $v\in C\in\C$.
\begin{algorithm}
\caption{\algo{Restrict-Demands}\label{alg:prize-restrict}}
\textbf{Input:} Instance $(G, \D, \pi)$ of \prob{Submodular Prize-Collecting Steiner Forest}\\
\textbf{Output:} Forest $F$ and $\dunsat$.

\begin{algorithmic}[1]
 \STATE Use the algorithm of Hajiaghayi et al.~\cite{HKKN10} to find a $3$-approximate 
solution:  a forest $F^+$ satisfying subset $\D^+$ of demands.
% \STATE Let $T^+_1, \dots, T^+_k$ denote the connected components of $F^+$.
 \STATE Construct $G^\ast(V, E^\ast)$ in which $E^\ast$ is the same as $E$ except that
the edges of $F^+$ have length zero in $E^\ast$. % Contract trees $T^+_i$ to construct the new graph $G^\ast(V^\ast,E^\ast)$.
% \STATE Let $\D^\ast$ be the contracted demand set:
%$i.e., $\{s',t'\} \in \D^\ast$ if and only if $s',t' \in V^\ast$ are the super-nodes corresponding
%to a demand $\{s,t\} \in \D$. 
 \STATE Define $\pi^\ast$ as Equation~\eqref{eqn:new-pi}.
% \STATE Let $\pi^\ast$ be the corresponding penalty function multiplied by $\eps^{-1}$.
% \STATE For any $v\in V$, let $\phi_v$ be $\frac{1}{\eps}$ times the cost of the tree $T^\ast_i$ corresponding to $v$, and zero if there is no such tree. .
 \STATE Call \algo{Submodular-PC-Clustering} on $(G^\ast, \D, \pi^\ast)$ to obtain the result $F$, $\dunsat$ and $y$.
% \STATE Uncontract $T^+_i$ in $F^\ast$ to obtain $F$.
% \STATE Find $\D'$ and $y$ in the same manner.
 \STATE Output $F$ and $\dunsat$.
%\STATE Output the set of trees $\{T_i\}$, along with $\D_i := \{ D\in\D: D\subseteq V(T_i)\}$.
\end{algorithmic}
\end{algorithm}

Now we show that the algorithm \algo{Restrict-Demands} outlined above  satisfies the requirements
of Theorem~\ref{thm:reduce}.  Before doing so, we show how the cost of
a forest can be compared to the values of the output vector $y$.
\begin{lemma}\label{lem:coloring}
If a graph $F$ satisfies a set $\dsat$ of demands,
then $\len(F) \geq \sum_{d\in\dsat} y_d$.
\end{lemma}
This is quite intuitive.
Recall that the $y$ variables color the edges of the graph.
Consider a segment on edges corresponding to cluster $S$ with color $d$.
%For any portion $(S,v)$ corresponding to the color $v$,
At least one edge of $F$ \emph{passes through} the cut $(S,\bar{S})$.
Thus, a portion of the cost of $F$ can be charged to $y_{S,d}$.
Hence, the total cost of the graph $F$ is at least as large as 
the total amount of colors paid for by $\dsat$.
We now provide a formal proof.
\begin{proof}
The length of the graph $F$ is
\begin{align*}
 \sum_{e\in F} c_e
  &\geq  \sum_{e\in F}\sum_{S: e\in\delta(S)} y_S &\text{by \eqref{eqn:s-lp:1}}\\
  &=     \sum_S |F \cap \delta(S)|y_S \\
  &\geq  \sum_{S: F \cap \delta(S)\neq\emptyset} y_S \\
  &=     \sum_{S: F\cap\delta(S)\neq\emptyset}\sum_{d: d\odot S} y_{S,d} \\
  &=     \sum_{d}\sum_{\substack{S: d\odot S\\F\cap\delta(S)\neq\emptyset}} y_{S,d} \\
  &\geq  \sum_{d\in \dsat}\sum_{\substack{S: d\odot S\\F\cap\delta(S)\neq\emptyset}} y_{S,d} \\
  &=     \sum_{d\in \dsat}\sum_{S: d\odot S} y_{S,d}, \\\intertext{because $y_{S,d}=0$ if $d\in \dsat$ and $F\cap\delta(S)=\emptyset$,}
  &=     \sum_{d\in \dsat} y_d  &\qedhere
\end{align*}
\end{proof}

\begin{proof}[\proofname\ of Theorem~\ref{thm:reduce}]
  We know that $\len(F^+) + \pi(\D \setminus \D^+) \leq 3\OPT$ because we
start with a $3$-approximate solution.
%Since $\D^\ast$ corresponds to the unsatisfied demands of $F^+$,
%we have $\pi^\ast(\D^\ast) = \eps^{-1}\pi(\D \setminus \D^+)$.
For any demand $d=(s,t)$, we know that $y_d$ is not more than the distance of $s, t$ in $G^\ast$.
Since distance between endpoints of $d$ is zero if it is satisfied in $\D^+$,
$y_d$ is non-zero only if $d\in\D\setminus\D^+$, we have $y(\D) = y(\D\setminus\D^+)\leq \pi^\ast(\D\setminus\D^+)$ by constraint~\eqref{eqn:s-lp:2}.
  Lemma~\ref{lem:submod-growth} gives $\len(F)$ in $G^\ast$, denoted by $\len_{G^\ast}(F)$,
is at most $2 y(\D) \leq 2\pi^\ast(\D\setminus\D^+) = 2\eps^{-1}\pi(\D\setminus\D^+) \leq 6\eps^{-1}\OPT$.
  Therefore, $\len(F) = \len(F^+) + \len_{G^\ast}(F) \leq (6\eps^{-1} + 3)\OPT$. 

To establish the second condition of the theorem, take an optimal forest $F'$:
$F'$ satisfies demands $\D^\OPT$, and we have $\len(F') + \pi(\D \setminus \D^\OPT) = \OPT$.
Define $A := \D^\OPT \setminus \dsat$ and $B := \dunsat \setminus A$.
%The additional penalty for $F'$ (i.e., if it pays the penalty for the entire set $\D \setminus \D'$)
The penalty of $F'$ under $\pi'$ is $\pi((\D\setminus\D^\OPT) \cup \dunsat) = \pi((\dsat\setminus\D^\OPT) \cup A \cup B)$.
Hence, the increase in penalty of $F'$ due to changing from $\pi$ to $\pi'$
is $\pi((\dsat\setminus\D^\OPT) \cup A \cup B) - \pi((\dsat\setminus\D^\OPT)  \cup B) \leq \pi(A\cup B) - \pi(B)$ due to the decreasing marginal cost property of submodular functions. 
%is at most $\pi(A \cup B) - \pi(B)$.
We have $y(A \cup B) = \pi^\ast(A \cup B) = \eps^{-1}\pi(A\cup B)$ because $A\cup B = \dunsat$ is the set of dead demands of \algo{Submodular-PC-Clustering}; see the first condition of Lemma~\ref{lem:submod-growth}.
We also have $\eps^{-1}\pi(B) = \pi^\ast(B) \geq y(B)$ because of Constraint \eqref{eqn:s-lp:2}.
Therefore, the additional penalty is at most $\eps [ y(A \cup B) - y(B) ] = \eps y(A)$.
Since $F'$ satisfies the demands $A$,
we have $y(A) \leq \len(F') \leq \OPT$ from Lemma~\ref{lem:coloring}.
Therefore, the additional penalty is at most $\eps \OPT$.
%Notice that $\pi(B) \leq \pi(\D \setminus \D^\OPT) \leq \OPT$.
%Thus, if $\pi(A \cup B) - \pi(B) \leq \eps \pi(B)$ then the claim follows.
%Let us assume this does not hold.
%Since $F'$ satisfies the demands $A$,
% we have $y(A) \leq \len(F') \leq \OPT$ from Claim~\ref{clm:coloring}.
%Notice that $y(B) = y(A \cup B) - y(A) \geq \eps^{-2} [\pi(A \cup B) - \pi(A)] \geq \eps^{-2} \cdot \eps \pi(A) \leq \eps \OPT$ as desired.

The extension to \spctsp and \spcs is straight-forward once we observe
that the cost of building a tour or a stroll on a subset $S$ of vertices
is at least the cost of constructing a Steiner tree on the same set.
Hence, there algorithm pretends it has an \prob{SPCST} instance, and restricts
the demand set accordingly.  However, the extra penalty due to the
ignored demands $\dunsat$ is charged to the Steiner tree cost which
is no more than the TSP or stroll length.
\end{proof}

%%%%%%%%%%%%%%%%%%%%%%%%
\subsection{Restricting the connectivity}\label{sec:reduce:2}

We first run \algo{Restrict-Demands} on $(G,\D,\pi)$.  
Let $F$ and $\dunsat$ be its output.
The forest $F$ satisfies all the demands in $\dsat := \D \setminus \dunsat$.
The length of this forest is $O(\OPT)$ and the demands
in $\dunsat$ can be safely ignored.

The forest $F$ consists of tree components $T_i$.
In the following, we connect some of these components to make the trees $\hat T_i$.
It is easy to see that this construction  guarantees the first two conditions of Theorem~\ref{thm:break}.
We work on a graph $G^\ast(V^\ast,E^\ast)$ formed from $G$ by contracting each tree component of $F$.
A potential $\phi_v$ is associated with each vertex $v$ of $G^\ast$,
which is $\eps^{-1}$ times the length of the tree component corresponding 
to $v$ in case $v$ is the contraction of a tree component, and zero otherwise.

We use the algorithm \algo{PC-Clustering} introduced in \cite{BHM10} to
cluster the components $T_i$ and construct a forest $F_2$
with components $\hat T_i$; the details of the algorithm can be seen
in \cite{BHM10}.  We obtain the folowing guarantees.

We first show the cost of the new edges is small.
\begin{lemma}[{\cite[Lemma 6]{BHM10}}]\label{lem:break:cost}
 The cost of $F_2$ is at most $2\sum_{v\in V^\ast} \phi_v$.
\end{lemma}

Recall that the trees $T_i$ are contracted in $F_2$. 
 Construct $\hat F$ from $F_2$ by uncontracting all these trees. 
 Let $\hat F$ consist of tree components $\hat T_i$.
 It is not difficult to verify that $\hat F$ is indeed a forest, but
 we do not need this condition since we can always remove cycles to
 find a forest.
% and let $W := \cup_{(s,t)\in\D}\{s,t\}$ denote the set of all terminals in $G$.
 Define $\hat \D_i := \{ (s,t)\in\D: s,t\in V(\hat T_i)\}$, and 
let $\D^\ast$ be the subset of demands satisfied by $\OPT$.
  Define $\D^\ast_i := \D^\ast \cap \D_i$, and denote by $\sfor(G, \D)$ the
length of a minimum Steiner forest of $G$ satisfying the demands $\D$.

\begin{lemma}[{\cite[Lemma 10]{BHM10}}]\label{lem:break:sep}
 $\sum_i\sfor(G, \D^\ast_i) \leq (1+\eps) \sfor(G, \D^\ast)$.
\end{lemma}

Now, we are ready to prove the main theorem of this section.
\begin{proof}[\proofname\ of Theorem~\ref{thm:break}]
%\xxx{edit this proof and statement of the theorem}
 The first condition of the lemma follows directly from our construction:  
we start with a solution, and never disconnect one of the tree components in the process.
 The construction immediately implies the second condition.
 By Lemma~\ref{lem:break:cost}, the cost of $F_2$ is at most $2\sum_{v\in V}\phi_v \leq \frac{2}{\eps}\len(F)$. Thus, $\hat F$ costs no more than $(2/\eps+1)\len(F)$, giving the third condition.
 Finally, Lemma~\ref{lem:break:sep} establishes the last condition.
\end{proof}

%%%%%%%%%%%%%%%%%%%%%%%%%%
%\subsection{The spanner and the reduction}
%
%\xxx{spanner ideas}
%
%\xxx{final reduction theorem -- general!}
%
%There are two PC-Clustering algorithms:
%1--each component has a potential of its own;
%2--each demand set has a fixed potential.

\section{PTASs for \pcst, \pctsp and \pcs on planar graphs}\label{sec:planar}
Since \pcst is a special case of \pcsf,
Theorems~\ref{thm:reduce-submod} and \ref{thm:btw} imply that
\pcst admits a PTAS on planar graphs.
However, obtaining the same result for \pctsp and \pcs is not immediate
from those theorems since the latter problems are not special
cases of \pcsf.
Here we explain how we can use these theorems to obtain the desired PTASs.
\iffalse

Let $\OPT^\pcst(\I)$, $\OPT^\pctsp(\I)$ and $\OPT^\pcs(\I)$ 
denote (the cost of) a minimum prize-collecting Steiner tree, tour and stroll %spanning
for an instance $\I = (G, \D, \pi)$, respectively.
%a subset $S$ of vertices (with associated penalties $\pi$) in $G$, respectively.
We know that $\OPT^\pcst(\I) \leq \OPT^\pctsp(\I) \leq 2\OPT^\pcst(\I)$
where the latter inequality is proved by observing that 
we can double the edges in a Steiner tree, find an Eulerian tour in it and then shortcut
to obtain a simple tour spanning the same set of vertices.
%We similarly have $\OPT^\pcst_S(G) \leq \OPT^\pcs_S(G) \leq 2\OPT^\pcst_S(G)$.
\fi
%In the following discussion 
Here we focus on $\pctsp$,
however, the same arguments with minor changes apply to $\pcs$ as well.

Take an instance $\I = (G, \D, \pi)$ of \pctsp,
and apply Theorem~\ref{thm:reduce} on $\I$ to
obtain $F$ and $\dunsat$.  Since all the demands
share a common root vertex\footnote{If we have a penalty for each vertex in
the \pctsp formulation, we can guess a root vertex $r$ and define the demand pairs accordingly.}, all the terminals
in $\dsat$ are connected in $F$.  We then invoke
the TSP spanner construction of Arora et al.~\cite{AGK98}
to build $H$.  Finally, we use the contraction decomposition
theorem of Demaine et al.~\cite{DHM07} to contract
a small-weight subset of edges and reduce the problem
to graphs of bounded treewidth.  The total additional charge
due to penalties of $\dunsat$ and contracted edges is
at most $O(\eps)\OPT$.  Therefore, we can obtain a PTAS
by solving the bounded-treewidth instance precisely.

\iffalse
We apply Theorem~\ref{thm:reduce} on $\I$
as if this is a \pcsf instance, and obtain $F$ and $\D'$.
The theorem implies that the additional cost due to
paying the penalties of $\D\setminus\D'$
is only $\eps\OPT^\pcsf(\I) \leq \eps\OPT^\pctsp(\I)$.
\xxx{it's better to change the statement of the theorem and its proof
to address \pctsp and \pcs directly.}
Therefore, we can focus on the remaining set of demands, $\D'$.
The goal is now to find a tour for a subset of $\D'$ and 
pay the penalty for the rest (in addition to $\D\setminus\D'$).
In $\pctsp$ all the demands pairs share a root vertex.\footnote{If we have a penalty for each vertex in
the \pctsp formulation, we can guess a root vertex $r$ and define the demand pairs accordingly.}
Since $F$ is a constant-approximate Steiner tree 
spanning all these terminals, we can use the results of \cite{AGK98}
to build a spanner for them, and then use the reduction idea based
on \cite{DHM07} to reduce it to bounded-treewidth graphs.
Finally, we can solve \pctsp on the bounded-treewidth instance to finish the proof.
\fi

\section{Hardness of \pcsf on series-parallel graphs}\label{sec:hard}
We first present the hardness proof for \pcsf on
a planar graph of treewidth two. The proof shows hardness for a very
restricted class of graphs: short cycles going through a single
central vertex.

\begin{proof}[\proofname\ of Theorem~\ref{thm:sf-hard}(1)]
  We reduce an instance $\I$ of \prob{Vertex Cover} on $3$-regular graphs
to an instance $\I'$ of \pcsf on a planar graphs of treewidth two.
  The former is known to be APX-hard \cite{MR2001c:90110}.
  The instance $\I$ is defined by an undirected graph $G$.
If $n$ denotes the number of vertices of $G$,
the number edges is $m = 3n/2$.
We will denote the $i$-th vertex of $G$ by $v_i$, the $j$-th edge by $e_j$, 
and the first and second endpoints of $e_j$
by $e^{(1)}_j$ and $e^{(2)}_j$, respectively.

 We now specify the reduction (illustrated in Figure~\ref{fig:hardness}); $\I'$ is represented by $(H, \D, \pi)$.
The graph $H$ consists of the vertices
\begin{itemize}
\item $a_i$ for $1\leq i\leq n$, 
\item $b_j, c^1_j , c^2_j$ for $1 \leq j \leq m$, 
\item central vertex $w$,
\end{itemize}
and the edges 
\begin{itemize}
\item $\{w,a_i\}$ of cost $2$ ($1\leq i\leq n$), 
\item $\{w,c^1_j\}, \{w,c^2_j\}, \{c^1_j,b_j\}, \{c^2_j,b_j\}$ of cost $1$ ($1 \leq j \leq m$).
\end{itemize}
The instance contains the following demands: 
\begin{itemize}
\item $\{w,b_j\}$ with penalty $3$ ($1\leq j\leq m$),
\item If $v_i =e^{(\ell)}_j$ for some $1\leq i\leq n$, $1\leq j\leq m$, and $\ell\in\{1,2\}$, then $\{a_i,c^\ell_j\}$ is a demand with penalty $1$.
\end{itemize}
Thus the number of demands is exactly $m + 3n$ and each $a_i$ appears 
in exactly $3$ demands. We claim that the cost of the optimum solution of 
$\I'$ is exactly $2m + 2n + \tau(G)$, where $\tau(G)$ is the size of the 
minimum vertex cover in $G$. Note that $\tau(G) \geq n/3$ (as $G$ is $3$-regular), 
thus $2m + 2n + \tau(G)$ is at most a constant times $\tau(G)$. In order to prove 
the correctness of the reduction, we prove the following two statements:
\begin{description}
\item[(1)] Given a vertex cover of size $k$ for $G$, 
a solution of cost $2m + 2n + k$ can be constructed.
\item[(2)]  Given a solution of cost at most $2m+2n+k$, a vertex cover of size at most $k$  
can be constructed.
\end{description}
To prove (1), suppose that $C$ is a vertex cover of size $k$ for $G$.
Let $T$ be a tree of $H$ that contains
\begin{itemize}
\item edge $\{w,a_i\}$ if and only if $v_i \not\in C$,
\item edges $\{w,c^1_j\}, \{c^1_j,b_j\}$ if and only if $e^1_j \not\in C$, 
\item edges $\{w,c^2_j\}, \{c^2_j,b_j\}$ if and only if $e^1_j \in C$.
\end{itemize}
\begin{figure}
\centerline{\includegraphics[height=8cm]{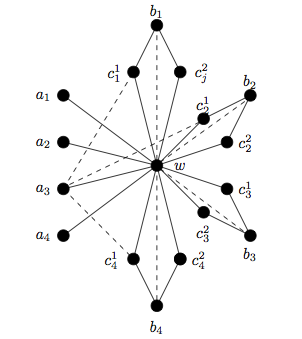}}
\caption{Illustrating the reduction from \prob{$3$-Regular Vertex Cover} to \pcsf.\label{fig:hardness}}
\end{figure}

The total cost of $T$ is $2(n - k) + 2m$. Observe that all the demands $\{w, b_j\}$ 
are connected (either via $c^1_j$ or $c^2_j$).
Furthermore,if $v_i\not\in C$, then all three demands where $a_i$ appears are satisfied: 
edge $\{w,a_i\}$ is in $T$ and if $v_i = e^1_j$, then edge $\{w,c^1_j\}$ is in $T$ as well. 
(Note that if $v_i = e^2_j$ and $v_i \not\in C$, then $e^1_j \in C$ must hold, 
and therefore $\{w, c^2_j\}$ is in $T$.) Thus the total penalty is at most $3k$, and 
hence the cost of the solution is at most $2n + 2m + k$, as claimed.

To prove (2), suppose that subgraph $F$ of $G$ is a solution such that 
the sum of the cost of $F$ and the penalties is at most $2m+2n+k$. 
We can assume that for every $1 \leq i \leq n$, vertex $b_j$ can be 
reached from $w$: otherwise we can decrease the penalty by $3$ 
at the cost of adding two edges of cost $1$. Furthermore, we can 
assume that only one of $c^1_j$ and $c^2_j$ is can be reached from $w$: 
otherwise we can remove an edge without disconnecting $b_j$ from $w$, 
thus the cost decreases by $1$ and the penalty increases by at most $1$. 
Finally, we can assume that if $\{w,a_i\} \in F$, then all $3$ demands 
containing $a_i$ are connected: otherwise removing $\{w,a_i\}$ decreases the cost
 by $2$ and increases the penalty by at most $2$.

Let vertex $v_i$ be in $C$ if and only if $\{w,a_i\} \not\in F$. 
We claim that $C$ is a vertex cover of size at most $k$. To see 
that $C$ is a vertex cover, consider an edge $e_j$. We have observed 
above that one of $c^1_j$ and $c^2_j$ cannot be reached from $w$. 
If $c^1_j$ cannot be reached from $w$ and $e^{(1)}_j = v_i$,
then the demand $\{v_i, c^1_j\}$ is not connected by $F$. 
Therefore, not all $3$ demands containing $a_i$ are connected, 
which means (as observed above) that $\{w,a_i\} \not\in F$. 
Thus $v_i \in C$, covering the edge $e_j$.

Since every $b_j$ can be reached from $w$ and $\{w,a_i\}\in F$ if $v_i\not\in C$,
the cost of $F$ is at least $2m + 2(n - |C|)$. Furthermore, if $v_i \in C$, then 
$\{w,a_i\} \not\in F$, which means that we have to pay the penalty for the $3$ demands
 containing $a_i$. Therefore, the total cost of the solution is at least $2m+2n+|C|$. 
We assumed that the cost of the solution is at most $2m+2n+|C|$, thus $|C| \leq k$
follows, what we had to prove.
\end{proof}
The proof for the Euclidean version is very similar to the graph
version. The main difference is that the central vertex $w$ is
replaced by a set of points arranged along a long vertical path.

\begin{proof}[\proofname\ of Theorem~\ref{thm:sf-hard}(2)]
  We reduce an instance $\I$ of \prob{Vertex Cover} on $3$-regular
  graphs to an instance $\I'$ of \pcsf on points in the Euclidean plane.  If $n$
  denotes the number of vertices of the 3-regular graph $G$ in $\I$,
  then the number edges is $m = 3n/2$.  We will denote the $i$-th vertex
  of $G$ by $v_i$, the $j$-th edge by $e_j$, and the first and second
  endpoints of $e_j$ by $e^{(1)}_j$ and $e^{(2)}_j$, respectively.

 We now specify the reduction (illustrated in
 Figure~\ref{fig:hardness2}). Let us define $U:=10000(n+m)$ (``basic unit
 of cost''), $H=10U$ (``horizontal length''), and $V=100U$ (``vertical spacing''). Instance $\I'$ contains the following
 set $P$ of points:
\begin{itemize}
\item $z_{0,y}=(0,y)$ for every $-mV\le y \le nV$,
\item $z_{x,y}=(x,y)$ and for every $0 \le x \le H$ and $y=iV$ for $1 \le
  i \le n$,
\item $z_{x,y}=(x,y)$ and $z_{x,y+4U}$ for every $0 \le x \le H$ and $y=-jV$ for $1 \le
  j \le m$,
\item $a_i=(H+2U,iV)$ for $1\leq i\leq n$, 
\item $b_j=(H,-jV+2U)$ for $1 \leq j \leq m$, 
\item $c^1_j=(H,-jV+U)$, and
  $c^2_j=(H,-jV+3U)$ for $1 \leq j \leq m$.
\end{itemize}
Let $Z$ be the set of all $z_{x,y}$ vertices in $P$, note that $|Z|=V(i+j)+1+(i+2j)H$. For ease of notation, we define $w_i=z_{H,iV}$, $w^1_j=z_{H,-jV}$,
$w^2_j=z_{H,-jV+4U}$.

The instance contains the following demands: 
\begin{enumerate}
\item If $z_{x,y}$ and $z_{x+1,y}$ are both in $P$, then there is a
  demand $\{z_{x,y},z_{x+1,y}\}$ with penalty 1.
\item If $z_{x,y}$ and $z_{x,y+1}$ are both in $P$, then there is a
  demand $\{z_{x,y},z_{x,y+1}\}$ with penalty 1.
\item $\{(0,0),b_j\}$ with penalty $3U$ ($1\leq j\leq n$),
\item If $v_i =e^{(\ell)}_j$ for some $1\leq i\leq n$, $1\leq j\leq m$,
  and $\ell\in\{1,2\}$, then $\{a_i,c^\ell_j\}$ is a demand with penalty
  $U-10$.
\end{enumerate}
The total
number of demands is $|Z|-1+n+3m$ and each $a_i$ appears in exactly $3$
demands. We claim that the cost of the optimum solution of $\I'$ is
between $|Z|+(2m+2n+\tau(G))U$ and $|Z|+(2m+2n+\tau(G))U-100n$, where
$\tau(G)$ is the size of the minimum vertex cover in $G$. Note that
$m=3n/2$ and $\tau(G)\ge m/3$, thus $|Z|+(2m+2n+\tau(G))U$ is at most
a constant factor larger than $\tau(G)U$.

More precisely, in order to prove the
correctness of the reduction, we prove the following two statements:
\begin{description}
\item[(1)] Given a vertex cover of size $k$ for $G$, 
a solution of cost at most $|Z|+(2m+2n+k)U$ for $\I'$ can be constructed.
\item[(2)]  Given a solution of cost at most $|Z|+(2m+2n+k)U$ for $\I'$, a vertex cover of size at most $k$  
can be constructed.
\end{description}
To prove (1), suppose that $C$ is a vertex cover of size $k$ for $G$.
Let $F$ be the forest (actually, a tree) that contains
\begin{enumerate}
\item edge $\{z_{x,y},z_{x+1,y}\}$ if both these points are in $P$,
\item edge $\{z_{x,y},z_{x,y+1}\}$ if both these points are in $P$,
\item edge $\{w_i,a_i\}$ if $v_i \not\in C$,
\item edges $\{w^1_j,c^1_j\}$ and $\{c^1_j,b_j\}$  if $e^{(1)}_j \not\in C$, 
\item edges $\{w^2_j,c^2_j\}$ and $\{c^2_j,b_j\}$ if $e^{(1)}_j \in C$.
\end{enumerate}

\begin{figure}
% \centerline{\includegraphics[height=8cm]{pcsf-hard-pic}}
\vskip-2cm
\centerline{\includegraphics[height=24cm]{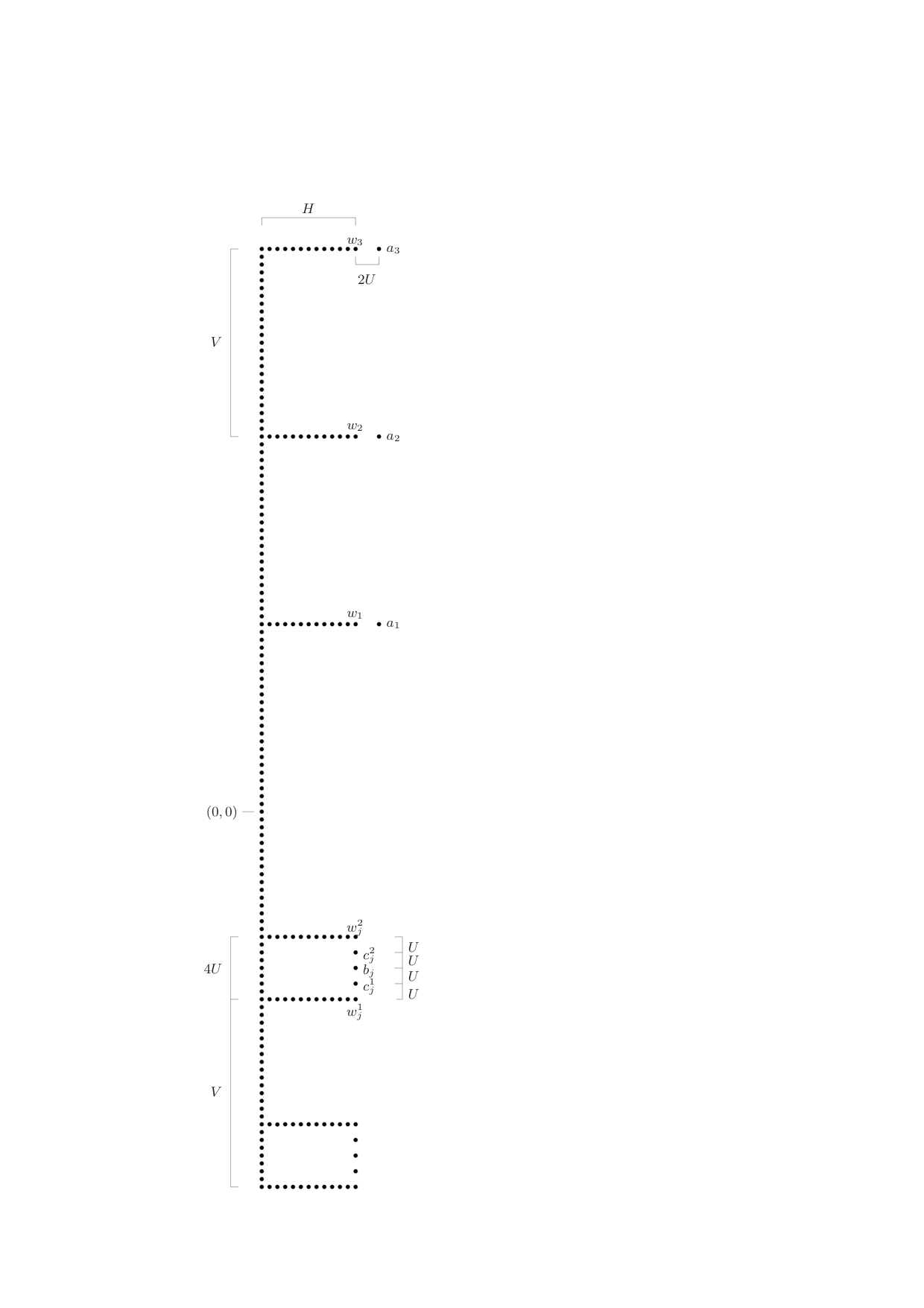}}
%\centerline{\input{pcsf2}}
\vskip-2cm
%\centerline{\special{ps: pcsf2.eps}}
 \caption{Illustrating the reduction from \prob{$3$-Regular Vertex
     Cover} to Euclidean \pcsf.\label{fig:hardness2}}
 \end{figure}

The total cost of $F$ is $|Z|-1+2U(n-k) + 2Um$. Observe that all the
demands $\{(0,0), b_j\}$ are satisfied.  Furthermore, if $v_i\not\in
C$, then all three demands where $a_i$ appears are satisfied. This can
be seen as follows. First, $a_i$ is in the same component as $w_i$ and
hence as every vertex of $Z$. If $v_i = e^{(1)}_j$, then there is a demand
$\{a_i,c^1_j\}$ and $c^1_j$ is connected with $w^1_j$ (and hence with
$a_i$). If $v_i = e^{(2)}_j$, then $v_i \not\in C$ means that $e^{(1)}_j \in
C$ must hold, and therefore $c^2_j$ is connected to $w^2_j$,
satisfying the demand $\{a_i,c^2_j\}$. Thus the total penalty is at most
$3k(U-10)$, and hence the cost of the solution is at most
$|Z|-1+(2m+2n+k)U-30k$, as claimed.

To prove (2), suppose that forest $F$ is an optimum solution such that
the sum of the cost of $F$ and the penalties is at most
$|Z|+(2n+2m+k)U$.  First, we can assume that every demand of the first
two types is satisfied: if, say, $(z_{x,y},z_{x+1,y})$ is not
satisfied, then we can extend $F$ by adding an edge of cost 1, which
decreases the penalty by at least 1. Thus all the $z_{x,y}$ points are
in the same connected component $K$ of $F$. We can also assume that
every demand of the third type is satisfied: if $\{(0,0),b_j\}$ is
not satisfied, then we can decrease the penalty by $3U$ at the cost of
$2U$ by adding edges $\{w^1_j,c^1_j\}$ and $\{c^1_j,b_j\}$,
contradicting the optimality of $F$. Therefore, every vertex $b_j$ is
in the component $K$.

Let
$Z'=\{z_{x,y}\in Z\mid x=0 \vee x\ge 10\}$. Let $R$ be the region of
the plane at Manhatten distance at most 3 from $Z'$. Note  that $R$
consists of one ``vertical'' and $n+2m$ ``horizontal'' components.

% We define the following subsets of points:
% \begin{itemize}
% \item $ZV_i=\{z_{x,iV}\mid 10 \le x \le W\}$ for $1\le i \le n$,
% \item $ZV_j=\{z_{x,-jV}\mid 10 \le x \le W\}$ for $1\le j \le m$,
% \item $ZV_i=\{z_{0,y}\mid iV-V-10 \le x \le iV-10\}$ for $1\le i \le n$,
% \item $ZV_j=\{z_{0,y}\mid jV-V-10 \le x \le jV-10\}$ for $1\le j \le m$.
% \end{itemize}
% We denote by $R(ZV_i)$ the region of the plane at Manhattan distance at most 3
% from the vertices of $ZV_i$ and similarly for the other
% subsets. Observe that the regions defined this way are pairwise
% disjoint. Let $R$ be the union of all these regions.

We claim that the cost of $F$ inside $R$ is at least $|Z'|$. We
have seen above that a single component $K$ of $F$ contains every
point of $P\cap R$. The restriction of $K$ to $R$ gives rise to
several components. Consider such a component $K'$ containing a subset
$S\subseteq Z'$ of vertices. We show that the cost of $K'$ is at least
$|S|$. The vertices of $S$ lie on a horizontal or vertical line.  This
means that there are two vertices $s_1,s_2\in S$ at distance $d\ge
|S|-1$. As $K$ is not contained fully in any component of $R$,
component $K'$ has to contain a point $s_3$ on the boundary of $R$. As
$s_3$ is at distance at least 3 from $s_1$ and $s_2$, it can be
verified that any Steiner tree of $s_1$, $s_2$, $s_3$ has cost at
least $d+1=|S|$. Summing for every component $K'$ of the restriction of $K$
to $R$, we get that the cost of $K$ in $R$ is at least $|P\cap R|$.

Let $R^+$ be the region of space at Manhattan distance at most 3 from
$Z$. We claim that the cost of every component of $F\setminus R^+$ is
at most $3U$. There are two types of components of $F\setminus R^+$:
(1) those that contain a point of $P$ and (2) those that do not
contain such a point. Clearly, there are at most $n+3m$ components of
the first type. Suppose that there is a component $D$ of the second
type having cost more than $3U$. In this case, we modify $F$ to
obtain a better solution as follows. Consider $F\setminus R^+$ (i.e.,
let us remove the part of $F$ inside $R^+$) and let
us remove every component of the second type. After that, let us add all the $|Z|-1$ edges
of the form $\{w_{x,y},w_{x+1,y}\}$, $\{w_{x,y},w_{x,y+1}\}$. Finally, for
  every component of the first type, if it intersects $R^+$, then let
  us choose a point of the component on the boundary of $R^+$ and
  connect this point to the nearest vertex of $Z$. It is clear that
  the new forest $F'$ satisfies every demand satisfied by $F$: every
  point of $P$ connected to $Z$ remains connected to $Z$. By our claim in the previous
  paragraph, the cost of $F\setminus R'$ is less than the cost of $F$
  by at least $|Z'|=|Z|-9(n+2m)$. Removing components of the second type decreases the
  cost by more than $3U$ (as there are at least one such component having
  cost more than $3U$). The edges connecting $Z$ increase the cost
  by $|Z|-1$. Adding the new connections corresponding to the
  components of the first type increases the cost by at most $n+3m$.
  As $3U\ge 9(n+2m)-1+n+3m$, forest $F'$ is a strictly better solution, a contradiction.

  Suppose now that there is a component $D$ of the first type with
  cost more than $3U$. For $-m \le s \le n$, let $R_s$ be the region
  of the plane at Manhattan distance at most $4U$ from $(H,sV)$.
  Observe that for each $s$, all the points of $P\cap R_s$ can be connected to the
  nearest point of $Z$ with a total cost of at most $3U$. This means
  that if $D$ intersects only one of these regions, say $R_s$, then we
  can substitute $D$ at cost at most $3U$ in such a way that every
  demand satisfied by $F$ remains satisfied, contradicting the
  optimality of $F$. Suppose therefore that $D$ intersects $t\ge 2$ of
  these regions; in this case, the cost of $D$ is at least
  $(t-1)(V-8U)>6tU-6U\ge 3tU$. Let us replace $D$ by connecting every
  point of $P\cap D$ to the closest vertex of $Z$. The new
  connections increase the cost by at most $t\cdot 3U$, which is less than
  the cost of $D$, a contradiction.

  We have proved that for every component $D$ of $F\setminus R^+$,
  $D\cap P$ is either a single $a_i$, or a subset of
  $\{b_j,c^1_j,c^2_j\}$. Therefore, every such component $D$ intersects
  $R^+$: otherwise, $D$could be safely removed, as it does not satisfy
  any demand. Next we show that it can be asssumed that
  only one of $c^1_j$ and $c^2_j$ is in $K$.  Otherwise we can remove
  every component of $F\setminus R^+$ intersecting
  $\{b_j,c^1_j,c^2_j\}$ and replace them with the edges
  $\{w^1_j,c^1_j\}$ and $\{c^1_j,b_j\}$. The total cost of the
  components we removed is at least $2U-3+U-3$ (which is the minimum
  cost of connecting $b_j$, $c^1_j$, $c^2_j$ to $R^+$) and the new
  edges have cost $2U$. This transformation might disconnect the
  demand containing $c^2_j$, hence the penalty can increase by at most
  $U-10$ only, contradicting the optimality of $F$. 

We can assume that if $a_i$ is in $K$, then all $3$ demands 
containing $a_i$ are connected: otherwise removing the component of
$F\setminus R^+$ containing $a_i$ decreases the cost
 by at least $2U-3$ and increases the penalty by at most $2(U-10)$.

 Let vertex $v_i$ be in $C$ if and only if $a_i$ is not in  component $K$.
 We claim that $C$ is a vertex cover of size at most $k$. To
 see that $C$ is a vertex cover, consider an edge $e_j$. We have
 observed above that one of $c^1_j$ and $c^2_j$ is not in $K$. If
 $c^1_j\not\in K$ and $e^{(1)}_j = v_i$,
 then the demand $\{a_i, c^1_j\}$ is not connected by $F$.  Therefore,
 not all $3$ demands containing $a_i$ are connected, which means (as
 observed above) that $a_i$ is not in $K$.  Thus $v_i \in C$,
 covering the edge $e_j$. Similarly, $c^2_j\not\in K$, then
 $e^{(2)}_j\in C$.

The cost of $F\cap R^+$ is at least $|Z|-9(n+2m)$. 
Since every $b_j$ is in $K$ and $a_i$ is in $K$ if $v_i\not\in C$,
the cost of $F\setminus R^+$ is at least $(2U-3)m + (2U-3)(n -|C|)$. Furthermore, if $v_i \in C$, then 
we have to pay the penalty for the $3$ demands
 containing $a_i$. Therefore, the total cost of the solution is at
 least 
\begin{multline*}
|Z|-9(n+2m)+(2U-3)m + (2U-3)(n -|C|)+3|C|(U-10)
\ge |Z|+(2m+2n+|C|)U-100n.
\end{multline*}
We assumed that the cost of the solution is at most $|Z|+(2m+2n+k)U$.
As $U> 100n$, this is only possible if $|C| \leq k$, what we had to
prove.
\end{proof}

%\section{Conclusion}

%\xxx{unify bib files; there are three double refs now.}
\bibliographystyle{siam}
\bibliography{main} %pcst,main2,penalty}

\appendix

%%%%%%%%%%%%%%%%%%%%%%%%%
\section{Basic graph theory definitions}\label{app:btw-defs}
Let $G(V,E)$ be a graph. As is customary, let $\delta(V')$ denote
the set of edges having one endpoint in a subset $V'\subseteq V$ of
vertices. For a subset of vertices $V'\subseteq V$, the subgraph of
$G$ induced by $V'$ is denoted by $G[V']$. With slight abuse of
notation, we sometimes use the edge set to refer to the graph
itself. Hence, the above-mentioned subgraph may also be referred to
by $E[V']$ for simplicity. We denote the length of a shortest
$x$-to-$y$ path in $G$ as $\dist_G(x, y)$. For an edge set $E$, we
denote by $\ell(E):=\sum_{e\in E}c_e$ the total length of edges in
$E$.

Given an edge $e=(u,v)$ in a graph $G$, the \emph{contraction} of
$e$ in $G$ denoted by $G/e$ is the result of unifying vertices $u$
and $v$ in $G$, and removing all loops and multiple edges except the
shortest edge.  More formally, the contracted graph $G/e$ is
formed by the replacement of $u$ and $v$ with a
 single vertex such that edges incident to the new vertex are the edges other than $e$ that
  were incident with $u$ or $v$. To obtain a simple graph, we first remove all self-loops
  in the resulting graph.
In case of multiple edges, we only keep the shortest edge and remove
all the rest. The contraction $G/E'$ is defined as the result of
iteratively contracting all the edges of $E'$ in $G$, i.e., $G/E' :=
G/e_1/e_2/\dots/e_k$ if $E'=\{e_1, e_2, \dots, e_k\}$. Clearly, the
planarity of $G$ is preserved after the contraction. Similarly,
contracting edges does not increase the cost of an optimal Steiner
forest.

  The boundary of a face of a planar embedded graph is the set of edges adjacent to the face; it does not always form a simple cycle. The boundary
$\partial H$ of a planar embedded graph $H$ is the set of edges bounding the infinite face. An edge is strictly enclosed by the boundary of $H$ if the edge belongs to $H$ but not to $\partial H$.

Now we define the basic notion of treewidth, as introduced
by Robertson and Seymour~\cite{RS86}.  To define this notion, we
consider representing a graph by a tree structure, called a tree
decomposition. More precisely, a \emph{tree decomposition} of a
graph $G(V,E)$ is a pair $(T,\B)$ in which $T(I,F)$ is a tree and
$\B=\{B_i\:|\:i\in I\}$ is a family of subsets of $V(G)$ such that
1) $\bigcup_{i\in I}B_i = V$; 2) for each edge $e=(u,v)\in E$, there
exists an $i\in I$ such that both $u$ and $v$ belong to $B_i$; and
3) for every $v\in V$, the set of nodes $\{i\in I\:|\:v\in B_i\}$
forms a connected subtree of $T$.

To distinguish between vertices of the original graph $G$ and
vertices of $T$ in the tree decomposition, we call vertices of $T$
\emph{nodes} and their corresponding $B_i$'s bags.  The \emph{width}
of the tree decomposition is the maximum size of a bag in $\B$ minus
$1$.  The \emph{treewidth} of a graph $G$, denoted $\tw(G)$, is the
minimum width over all possible tree decompositions of $G$.

For algorithmic purposes, it is convenient to define a restricted
form of tree decomposition. We say that a tree decomposition
$(T,\B)$  is {\em nice} if the tree $T$ is a rooted tree such that
for every $i\in I$ either
\begin{enumerate}
\item $i$ has no children ($i$ is a {\em leaf node}),
\item $i$ has exactly two children $i_1$, $i_2$ and
  $B_i=B_{i_1}=B_{i_2}$ holds ($i$ is a {\em join node}),
\item $i$ has a single child $i'$ and $B_i=B_{i'}\cup\{v\}$ for
  some $v\in V$ ($i$ is an {\em introduce node}), or
\item $i$ has a single child $i'$ and $B_i=B_{i'}\setminus \{v\}$ for
  some $v\in V$ ($i$ is a {\em forget node}).
\end{enumerate}
It is well-known that every tree decomposition can be transformed
into a nice tree decomposition of the same width in polynomial time.
Furthermore, we can assume that the root bag contains only a single vertex.

We also need a basic notion of embedding; see, e.g., \cite{RS94,
CM05}. In this paper, an \emph{embedding} refers to a \emph{$2$-cell
embedding}, i.e., a drawing of the vertices and edges of the graph
as points and arcs in a surface such that every face (connected
component obtained after removing edges and vertices of the embedded
graph) is homeomorphic to an open disk. We use basic terminology and
notions about embeddings as introduced in \cite{MT01}.  We only
consider compact surfaces without boundary.  Occasionally, we refer
to embeddings in the plane, when we actually mean embeddings in the
$2$-sphere.  If $S$ is a surface, then for a graph $G$ that is
($2$-cell) embedded in $S$ with $f$ facial walks, the number
$g=2-|V(G)|+|E(G)|-f$ is independent of $G$ and is called the
\emph{Euler genus} of $S$. The Euler genus coincides with the
crosscap number if $S$ is non-orientable, and equals twice the usual
genus if the surface $S$ is orientable.

\section{\pcst, \pctsp and \pcs on bounded-treewidth graphs}\label{app:btw}

Treewidth is a notion of how similar a graph is to trees.  Since tree structure
usually lends itself to the dynamic programming approach, it is plausible
that many optimization problems may be solvable in polynomial time
on graphs of bounded treewidth; Bodlaender and Koster~\cite{BK08} have a
comprehensive survey on this topic.
In particular, several Steiner network problems become relatively easy
when restricted to bounded-treewidth graphs.  Among them are
\prob{Steiner Tree}, \prob{TSP} and \prob{Stroll}.  One surprising outlier
is \prob{Steiner forest} that is proved to be NP-hard, yet it admits a PTAS~\cite{BHM10}.
In this section, we study the prize-collecting extensions of the above problems,
and when possible, we provide a polynomial-time algorithm for them.
More specifically, %in Section~\ref{sec:btw:pos}, 
we present PTASs for \pcst, \pctsp and \pcs on bounded-treewidth graphs.  
%Afterwards, in Section~\ref{sec:hard}, we show 
We already showed in Section~\ref{sec:hard}
that \pcsf is APX-hard
even on series-parallel graphs.  The proof is extended to give APX-hardness for Euclidean plane.

%\subsection{Positive results}
\label{sec:btw:pos}
We focus the discussion on \pcst, however, minor modifications
allow us to solve \pctsp and \pcs, too.
We are given a weighted graph $G(V,E)$ of treewidth $k-1$
for a fixed parameter $k$, and a penalty function $\pi: V\rightarrow
\R_+$.
We have a nice tree
decomposition $(T,\mathcal{B})$ for $G$. % where $T(I,F)$ is the
%tree structure.
Each bag $B_i$ has size at most $k$.  These are sometimes
called \emph{portals} for the subtree below node $B_i$.
Let $I$ denote the nodes of the tree decomposition $T$,
and for each $i\in I$, let $T_i$ be the subtree of $T$
below $i$.
A dynamic programming entry is specified by a tuple
$(i, S, \mathcal{P})$ where 
\begin{itemize}
\item $i\in I$ is a node in the tree decomposition, 
\item $S \subseteq B_i$ is a subset of portals of
the subtree $T_i$, and
\item $\mathcal{P}$ is a partition of $S$.
\end{itemize}
Let us denote by $V_i$ the vertices corresponding to the
subtree $T_i$, i.e., $V_i := \cup_{i'\in T_i} B_{i'}$.
A dynamic programming entry $\DP(i, S, \mathcal{P})$
takes up the least cost of building a subgraph $H$ such that
\begin{itemize}
\item $H$ uses only the edges whose both endpoints are in $V_i$,
\item $H$ connects the vertices in each set $P_j$ of the partition
$\mathcal{P} = \{P_1, P_2, \ldots, P_m\}$,
\item $S$ is the subset of $B_i$ whose penalty is not paid, moreover, if a vertex $v\in V_i$ is not connected to $S$ via $H$, then
its penalty $\pi(v)$ is paid in the total cost.
\end{itemize}

The final solution to the problem can be found as
$\min_S{\DP(r, S, \{S\})}$ where $r$ is the root of the tree
decomposition, i.e., it does not matter which subset of the bag
of the root is picked as long as they form a single component.

The DP entries are easy to compute for leaves:
let $B_i = \{v\}$ for a leaf $i$.
There are two possibilities:
$\DP(i, \emptyset, \emptyset) = \pi(v)$ and $\DP(i, \{v\}, \{\{v\}\})
= 0$.
The update procedure works as follows for different tree nodes:
\begin{description}
\item [Introduce node] $i$ is the parent of $i'$, and we have $B_i =
B_{i'} \cup \{v\}$.  Then, $\DP(i, S, \mathcal{P}) = \pi(v) + \DP(i',
S, \mathcal{P})$ if $v\not\in S$.  Next consider an entry $\DP(i, S,
\mathcal{P})$ such that
for $v\in S$ and $\mathcal{P} = \{P_1, P_2, \ldots, P_m\}$ where
$v\in P_1$.  Let $\mathcal{P}' := \{P_1\setminus\{v\}, P_2, \ldots,
P_m\}$ and let $d$ be the distance of $v$ to the set $P_1 \setminus
\{v\}$.  The dynamic programming sets $\DP(i, S, \mathcal{P}) =
d + \DP(i', S\setminus\{v\}, \mathcal{P}')$.
\item [Forget node]
$i$ is the parent of $i'$, and we have $B_{i'} =
B_{i} \cup \{v\}$.  Then, 
\begin{align*}
\DP(i, S, \mathcal{P}) = \min\Big[&\pi(v) + \DP(i',
S, \mathcal{P}), \\&\min_{\mathcal{P}'}\left\{\DP(i', S\cup\{v\},
\mathcal{P}') : \mathcal{P}'\mbox{ is formed by adding $v$ to a set of
}\mathcal{P}\right\}\Big].
\end{align*}
The first terms considers the case where we pay the penalty for $v$
and do not connect it in the final Steiner tree, whereas the
second term takes into account the case where $v$ is connected
to each connected component of the partition.
\item [Join node] the node $i$ has two children $i_1$ and $i_2$
with the same bags.  We set $\DP(i, S, \mathcal{P})$ to
$$\min_{\mathcal{P}_1, \mathcal{P}_2} \left\{\DP(i_1, S, \mathcal{P}) + \DP(i_2, S, \mathcal{P}) -
\pi(B_i\setminus S)\right\},$$
%$$\min_{\parbox{5cm}{$\mathcal{P}_1, \mathcal{P}_2:$ connectivity of
%$\mathcal{P}_1$ and $\mathcal{P}_2$
%partitions implies connectivity of }\mathcal{P}} \hspace{-2cm}\left\{\DP(i_1, S, \mathcal{P}) + \DP(i_2, S, \mathcal{P}) -
%\pi(B_i\setminus S)\right\},$$ 
%where the last term appears to cancel double
%charging the unsatisfied terminals.
where the minimization goes over all pairs $\mathcal{P}_1$ and
$\mathcal{P}_2$ whose connectivity implies that of $\mathcal{P}$.
The last term in the minimum operand is for canceling the
double charging of the unsatisfied terminals of $B_i$.
\end{description}

It is not difficult to verify that the algorithm produces the correct
output, and we defer the proof to the full version of the paper.
The running time of the algorithm is polynomial in the number of
DP entries, and the latter is at most $n\cdot 2^k\cdot k^k$.
Since $k$ is a constant, the running time is a polynomial.

To extend the algorithm to \pctsp, the DP state is modified
to $(i, \mathcal{P})$ where $i\in I$ is a node of the tree
decomposition, and $\mathcal{P}$ is a set of pairs of vertices
in bag $B_i$.  A pair $s,t$ implies that there is a path between
$s$ and $t$ in the subsolution, but the two nodes should be extended
from outside the subtree $T_i$ to make a tour.  The final solution
is stored in $\DP(r, \{(r,r)\})$.  The algorithm for \pcs works in the
same way except that the final solution can be founded in
$\min_{s,t\in B_r} \DP(r, \{(s,t)\})$ since we do not need to have
a closed tour.

\end{document}